\documentclass[twocolumn,envcountsect]{svjour3}
\def\makeheadbox{\relax}

\usepackage{ifluatex}
\ifluatex
  \usepackage{fontspec}
  \usepackage[english]{selnolig}
\fi

\ifluatex
\else
  \usepackage[T1]{fontenc}
  \usepackage[utf8]{inputenc} %
\fi

\usepackage{graphicx}

\usepackage{upquote}

\usepackage{booktabs}

\usepackage{paralist}

\usepackage{csquotes}

\usepackage{textcmds}

\usepackage{url}
\makeatletter
\g@addto@macro{\UrlBreaks}{\UrlOrds}
\makeatother
\usepackage{xcolor}

\usepackage{listings}
\renewcommand{\lstlistingname}{List.}
\usepackage{chngcntr}
\AtBeginDocument{\counterwithout{lstlisting}{section}}

\usepackage{pdfcomment}
\usepackage{stfloats}
\fnbelowfloat

\usepackage{hyperref}
\hypersetup{hidelinks,
  colorlinks=true,
  allcolors=black,
  pdfstartview=Fit,
  breaklinks=true}

\usepackage[all]{hypcap}

\usepackage[group-digits,per-mode=fraction]{siunitx}

\usepackage{amsmath}
\usepackage{amssymb}
\usepackage{mathtools}
\usepackage[ruled,linesnumbered,lined]{algorithm2e}

\smartqed

\usepackage[capitalise,nameinlink]{cleveref}
\usepackage{iflang}
\IfLanguageName{ngerman}{
  \crefname{table}{Tab.}{Tab.}
  \Crefname{table}{Tabelle}{Tabellen}
  \crefname{figure}{\figurename}{\figurename}
  \Crefname{figure}{Abbildungen}{Abbildungen}
  \crefname{equation}{Gleichung}{Gleichungen}
  \Crefname{equation}{Gleichung}{Gleichungen}
  \crefname{listing}{\lstlistingname}{\lstlistingname}
  \Crefname{listing}{Listing}{Listings}
  \crefname{section}{Abschnitt}{Abschnitte}
  \Crefname{section}{Abschnitt}{Abschnitte}
  \crefname{paragraph}{Abschnitt}{Abschnitte}
  \Crefname{paragraph}{Abschnitt}{Abschnitte}
  \crefname{subparagraph}{Abschnitt}{Abschnitte}
  \Crefname{subparagraph}{Abschnitt}{Abschnitte}
}{
  \crefname{section}{Sec.}{Sec.}
  \Crefname{section}{Section}{Sections}
  \crefname{listing}{\lstlistingname}{\lstlistingname}
  \Crefname{listing}{Listing}{Listings}
}
\crefname{enumi}{}{}
\Crefname{enumi}{}{}
\creflabelformat{enumi}{(#2#1#3)}
\crefname{definition}{Def.}{Def.}
\crefname{corollary}{Cor.}{Cor.}
\crefname{theorem}{Thm.}{Thm.}

\ifluatex
\else
  \input glyphtounicode
\fi

\usepackage{subcaption}
\usepackage{tikz}

\DeclareMathOperator{\V}{\mathbf{V}}
\DeclareMathOperator{\Q}{\mathbf{Q}}

\DeclareMathOperator{\qset}{qset}

\DeclareMathOperator{\outlinks}{edges^{+}}
\newcommand{\rank}{R}

\newcommand{\InnerQSets}{\mathcal{I}}
\newcommand{\Quorums}{\mathcal{U}}
\newcommand{\NodeSetSet}{\mathcal{N}}
\newcommand{\BlockSetSet}{\mathcal{B}}
\newcommand{\SplitSetSet}{\mathcal{S}}
\newcommand{\QSetSpace}{\mathfrak{D}} %

\newcommand{\Min}[1]{\hat{#1}}

\newcommand{\tts}{m}

\DeclarePairedDelimiter\set{\{}{\}}
\DeclarePairedDelimiter\len{\lvert}{\rvert}
\DeclarePairedDelimiter\ceil{\lceil}{\rceil}

\newcommand{\optithreshraw}[1]{\ceil{\frac{2#1+1}{3}}}
\newcommand{\optithresh}[1]{\optithreshraw{\len{#1}}}

\newcommand{\oldASG}{G_\text{AS98}}
\newcommand{\newASG}{G_\text{AS20}}

\begin{document}

\title{The Sum of Its Parts: Analysis of Federated Byzantine Agreement Systems}

\author{Martin Florian \and Sebastian Henningsen \and Charmaine Ndolo \and\\ Björn Scheuermann}

\institute{
  M. Florian, S. Henningsen, C. Ndolo \at Weizenbaum Institute / Humboldt University of Berlin
  \and
  B. Scheuermann \at Technical University of Darmstadt
}

\date{}

\maketitle

\begin{abstract}
Federated Byzantine Agreement Systems (FBASs) are a fascinating new paradigm in the context of consensus protocols.
Originally proposed for powering the Stellar payment network,
FBASs can instantiate Byzantine quorum systems without requiring out-of-band agreement on a common set of validators;
every node is free to decide for itself with whom it requires agreement.
Sybil-resistant and yet energy-efficient consensus protocols can therefore be built upon FBASs,
and the \enquote{decentrality} possible with the FBAS paradigm might be sufficient to reduce the use of environmentally unsustainable proof-of-work protocols.
In this paper, we first demonstrate how the robustness of individual FBASs can be determined,
by precisely determining their safety and liveness buffers and therefore enabling a comparison with threshold-based quorum systems.
Using simulations and example node configuration strategies,
we then empirically investigate the hypothesis that while FBASs can be bootstrapped in a bottom-up fashion from individual preferences,
strategic considerations should additionally be applied by node operators in order to arrive at FBASs that are robust and amenable to monitoring.
Finally,
we investigate the reported \enquote{open-membership} property of FBASs.
We observe that an often small group of nodes is exclusively relevant for determining liveness buffers
and prove that membership in this top tier is conditional on the approval by current top tier nodes if maintaining safety is a core requirement.
\end{abstract}


\section{Introduction}
\label{sec:intro}

We study \emph{Federated Byzantine Agreement Systems} (FBASs),
as originally proposed by Mazières~\cite{mazieres2015stellar}.
FBASs are conceptually related to
\emph{Asymmetric Quorum Systems}~\cite{cachin2019asymmetric}
and \emph{Personal Byzantine Quorum Systems}~\cite{losa2019stellar_instantiation}.
While research on consensus protocols has accelerated in the wake of global blockchain enthusiasm,
developments still mostly fall in two extreme categories:
\emph{permissionless}, i.e., open-membership,
as exemplified by Bitcoin's notoriously energy-hungry \enquote{Nakamoto consensus} \cite{nakamoto2008bitcoin},
and \emph{permissioned}, with a closed group of validators,
as assumed both in the classical Byzantine fault tolerance (BFT) literature (e.g., \cite{castro1999practical})
and many state-of-the art protocols from the blockchain world (e.g., \cite{yin2019hotstuff_podc}).
The FBAS paradigm and the works it has inspired suggest a middle way:
\emph{Each} node defines its \emph{own rules} about which groups of nodes it will consider as sufficient validators.
If the sum of all such configurations fulfills a set of properties,
protocols like the \emph{Stellar Consensus Protocol} (SCP)~\cite{mazieres2015stellar} can be defined
that leverage the resulting structure for establishing a live and safe consensus
system~\cite{cachin2020asymmetric,losa2019stellar_instantiation,garcia2019deconstructing,garcia2018fbqs,lokhava2019stellar_payments}.

In the original FBAS model~\cite{mazieres2015stellar},
which this paper is based on,
these properties are foremost \emph{quorum availability despite faulty nodes}, which enables \emph{liveness},
and \emph{quorum intersection despite faulty nodes},
which makes it possible for consensus protocols to prevent forks and thus enables \emph{safety}.
In a practical deployment, it is seldom clear which nodes are faulty,
and in this way the level of risk w.r.t. to liveness and safety is uncertain.
We propose an intuitive and yet precise analysis approach for determining the level of risk,
based on enumerating \emph{minimal blocking sets} and \emph{minimal splitting sets}---minimal sets of nodes that,
if faulty,
can by themselves compromise \emph{liveness} and \emph{safety}.
We provide algorithms for determining these sets in arbitrary FBASs and make available an efficient software-based analysis framework\footnote{
  \url{https://github.com/wiberlin/fbas_analyzer}
}.
To the best of our knowledge,
we are the first to propose and implement an analysis methodology for the assessment of the liveness and safety guarantees of FBAS instances
that yields precise results as opposed to heuristic estimations.
As previously shown in \cite{garcia2018fbqs},
FBASs induce \emph{Byzantine quorum systems} as per Malkhi and Reiter \cite{malkhi1998byzantine}---hence our results might be of interest to more classical formalizations as well.
For example,
we explicitly distinguish between sets of nodes that can undermine liveness and such sets that can undermine safety,
highlighting that in an actual system the threat to liveness and the threat to safety can differ both in structure and in severity.

We apply our analysis approach and tooling in an empirical study
that investigates the emergence of FBASs from existing inter-node relationships,
as encoded in, e.g., trust graphs.
Based on example \emph{configuration policies},
we demonstrate that while FBASs can be bootstrapped in a bottom-up fashion from individual preferences,
strategic considerations should additionally be applied by node operators in order to arrive at FBASs that are robust and amenable to monitoring.

Strategic considerations can increase centralization, on top of what is already implied by individual preferences.
We observe that centralization manifests as a \emph{top tier} of nodes that is solely relevant when determining liveness buffers.
We contribute a proof that if maintaining basic safety guarantees is a minimal strategic requirement of node operators,
top tiers are effectively \enquote{closed-membership} in the sense that a top tier's composition can only change with cooperation of current top tier nodes.
This casts doubt on the reported \enquote{open-membership} property of FBASs---while any node can become part of the FBAS,
our results show that only nodes approved by the current top tier can become relevant for consensus.

Following an overview of related work (\cref{sec:related_work}) and the formal introduction of the FBAS model and its interpretation in practical deployments (\cref{sec:fba}),
we structure our paper around our main original contributions:

\begin{itemize}
  \item An analysis framework for reasoning about safety and liveness guarantees in concrete FBASs (\cref{sec:analysis}).
  \item Algorithms for efficiently performing the proposed analyses (\cref{sec:algorithms}).
  \item A simulation-based exploration of possible configuration policies and their effects (\cref{sec:qsc}).
  \item Formal proof that membership in an FBAS' top tier is only \enquote{open} if a violation of safety is considered acceptable (\cref{sec:openness}).
\end{itemize}

As appendices, we prove a number of additional corollaries and theorems (\cref{app:proofs}) and present results from applying our analysis methodology to an interesting toy network (\cref{app:cascading_example}) and the current Stellar network (\cref{app:stellar_example}).

\section{Related work}%
\label{sec:related_work}

\emph{Federated Byzantine Agreements Systems} were first proposed in \cite{mazieres2015stellar},
together with the \emph{Stellar Consensus Protocol} (SCP),
a first protocol for this setting.
The viability of SCP has been proven formally~\cite{lokhava2019stellar_payments,garcia2019deconstructing,garcia2018fbqs}
and the protocol is in active use in two large-scale payment networks~\cite{lokhava2019stellar_payments,ndolo2021crawling}.
The FBAS notion has furthermore been generalized and reformulated in different ways,
creating bridges to more classical models and
enabling the development of additional protocols~\cite{losa2019stellar_instantiation,cachin2020asymmetric,cachin2019asymmetric}.
Among other things, as shown by Garc{\'i}a-P{\'e}rez and Gotsman \cite{garcia2018fbqs},
FBASs with \enquote{safe} configurations induce Byzantine quorum systems \cite{malkhi1998byzantine}.
In this work,
we are less interested in the mechanics of specific protocols for the FBAS setting
but instead investigate the \emph{conditions} they require for achieving safety, liveness and performance.
We investigate how many node failures (and of which nodes) an FBAS can tolerate before the conditions to safety and liveness are compromised,
and how individual node configuration policies influence these \enquote{buffers}.

Previously,
consensus protocols relevant in practice (such as PBFT~\cite{castro1999practical}) have relied on a symmetric threshold model.
In a typical instantiation with $3f+1$ nodes that can tolerated up to $f$ Byzantine node failures,
each $2f+1$ nodes form a (minimal) quorum.
This model naturally gives rise to quorum systems that are trivial to analyze,
i.e., for which it is trivial to determine under which maximal fail-prone sets \cite{malkhi1998byzantine} consensus is still possible.
The possibility for quorum systems that lack symmetry
(that is opened up by the FBAS paradigm and related notions)
makes the investigation of a more general analysis approach necessary.

A heuristics-based methodology for analyzing FBAS instances was previously proposed in \cite{kim2019stellar_secure},
focusing on the identification of central nodes and threats to FBAS liveness.
We propose a novel analysis approach that is not heuristics-based and hence yields precise insights,
based on a solid theoretic foundation.
As in \cite{kim2019stellar_secure}, we apply our methodology to snapshots of the live Stellar network (cf. \cref{app:stellar_example}).

Bracciali et al.~\cite{bracciali2021decentralization} explore fundamental bounds on the decentrality in open quorum systems.
One of their central arguments with regards to the FBAS paradigm is that quorum intersection,
a crucial requirement to guaranteeing safety in protocols like SCP,
is computationally intractable to determine and maintain,
necessitating centralization if safety is a requirement.
The NP-hardness of determining quorum intersection was previously also proven by Lachowski~\cite{lachowski2019complexity},
together, however,
with practical algorithms for nevertheless determining safety-critical properties of non-trivial FBASs.
We develop new algorithms that incorporate the possibility that some nodes may fail,
enumerating \emph{minimal blocking sets} and \emph{minimal splitting sets}.
We evaluate their performance for different FBAS sizes,
providing insights into the computational limitations that are relevant in practice.
While, based on our analysis approach and its application to specific FBASs,
we can confirm that nodes of higher influence
(\emph{top tier} nodes according to our choice of words)
naturally emerge,
we argue that it is not only the existence and size of such a group that determines \enquote{centralization}
but also the fluidity of that group's membership (which we explicitly investigate).

An alternative analysis methodology and software framework
has recently been presented in \cite{gaul2019mathematical}.
Among other things,
the authors provide algorithms for determining the consequences of specific sets of nodes becoming faulty,
whereas we propose and implement approaches for identifying \emph{all} minimal sets of nodes that need to become faulty for an FBAS to lose safety and liveness guarantees.

\section{Federated Byzantine agreement}%
\label{sec:fba}

In the following, we introduce core concepts of the FBAS paradigm that form our basis for reasoning about specific FBAS instances.
We use terminology based on \cite{lokhava2019stellar_payments}, \cite{mazieres2015stellar},
\cite{lachowski2019complexity} and the Stellar codebase (\emph{stellar-core}).

Our FBAS model is based on the concept of \emph{nodes}.
Whereas nodes usually represent individual machines,
for the purposes of this paper
we typically assume that each node represents a distinct entity or organization.
We will illustrate introduced concepts using examples,
with nodes represented as integers.
For example, $\set{0, 1, 2}$ denotes a set of three distinct nodes.
We will occasionally also use established terms in the context of consensus protocols,
such as \enquote{slot}, \enquote{externalize} and \enquote{faulty}, without formally introducing them.
As an informal and approximate adaptation to the blockchain setting,
a \emph{slot} is a block of a given height,
to \emph{externalize} a \emph{value} is to decide the contents of a block\footnote{
  Consensus protocols for the FBAS setting typically provide \emph{immediate finality},
  in the sense that once the value for a slot
  has been externalized,
  it cannot be reverted or changed.
},
and a \emph{faulty} node is one that violates protocol rules in arbitrary ways,
e.g., assuming the worst-case scenario,
via being under the control of an attacker that also controls all other faulty nodes.

We first introduce the formal foundation of the FBAS paradigm as originally proposed in \cite{mazieres2015stellar}.
Following that, we formally define the \emph{quorum set} configuration format for FBAS nodes that was previously only used in a practical implementation (of the Stellar network software)
but whose convenience for defining specific FBAS instances also benefits the theoretical discussion.
Based on the introduced foundations, we finally derive the necessary properties an FBAS must exhibit in order to enable liveness and safety guarantees.

\subsection{Quorum slice and FBAS}%
\label{sub:fbas}

In an FBAS, each node (respectively its human administrator) individually configures which other
nodes' opinions it should consider when participating in consensus.
Configurations can express individual expectations,
such as \enquote{out of these $n$ nodes, at most $f$ will simultaneously cooperate to attack the system},
and can be used to strategically influence global system parameters.
On a conceptual level, the configuration of an FBAS node consists in the definition of \emph{quorum slices}.

\begin{definition}[FBAS; adapted from \cite{mazieres2015stellar}]
  \label{def:fbas}
  A Federated Byzantine Agreement System (FBAS) is a pair
  $(\V, \Q)$
  comprising a set of nodes $\V$ and a quorum function
  $\Q : \V \to 2^{2^{\V}}$
  specifying \emph{quorum slices} for each node,
  where a node belongs to all of its own quorum slices---i.e.,
  $\forall v \in \V, \forall q \in \Q(v), v \in q$.
\end{definition}

Informally, each quorum slice of a node $v$ describes a set of nodes that,
should they all agree to externalize a value in a given slot,
is sufficient to also cause $v$ to externalize that value.

Clearly, an FBAS cannot be modeled as a regular graph (with FBAS nodes as graph edges) without losing information.
Graph-based analyses as in \cite{kim2019stellar_secure} can therefore result only in heuristic insights.
An FBAS \emph{can} be modeled as a directed hypergraph~\cite{gallo1993hypergraph}.
However, we find the \emph{quorum set} abstraction (presented next) more suitable for subsequent analysis.
In \cref{sec:qsc}, we explore strategies for bootstrapping robust FBASs \emph{from} graphs.

\subsection{Quorum set}%
\label{sub:quorum_set}

While a useful abstraction for formally describing protocols for the FBAS setting,
quorum slices are an unwieldy format for describing concrete FBAS instances.
In Stellar,
the currently most relevant practical deployment of an FBAS,
nodes are configured not via quorum slices but via \emph{quorum sets}~\cite{lokhava2019stellar_payments}.
Each quorum set defines
a set of validator nodes $U \subseteq \V$,
a set of inner quorum sets $\InnerQSets$ and a threshold value $t$.
Intuitively,
this representation enables the encoding of notions such as \enquote{out of these nodes $U$, at least $t$ must agree}
(\emph{satisfying} the quorum set)
or \enquote{the sum of agreeing nodes in $U$ and satisfied inner quorum sets in $\InnerQSets$ must be at least $t$}.

\begin{definition}[quorum set; adapted from Stellar codebase]
  \label{def:quorum_set}
  A quorum set is a recursive tuple
  $(U, \InnerQSets, t) \in \QSetSpace, \: \QSetSpace := 2^{\V} \times 2^\QSetSpace \times \mathbb{Z}^{+}$.
  For quorum sets of the form $D = (U, \InnerQSets, t)$,
  we recursively define that a set of nodes $q \subseteq \V$ \emph{satisfies} $D$ iff
  $(\len{q \cap U} + \len{\set{I \in \InnerQSets : q \text{ satisfies } I}}) \geq t$.
\end{definition}

For example, $(\set{0, 1},\emptyset, 1)$ encodes that agreement is required from either node $0$ or node $1$,
whereas $(\set{0}, \InnerQSets, 1)$ with $\InnerQSets = \set{(\set{1, 2, 3}, \emptyset, 2)}$ encodes that either node $0$ or two out of $\set{1, 2, 3}$ must agree.
Inner quorum sets (members of $\InnerQSets$) are often used for grouping nodes belonging to the same entity (respectively organization),
so that the importance of an entity can be decoupled from the number of nodes it controls.

Quorum sets are useful for defining the quorum slices of a node.
To ease notation, we define the
formalism $\qset(v, D)$ that expresses the set of quorum slices of a node $v \in \V$ based on a quorum set $D \in \QSetSpace$.

\begin{definition}[quorum set $\to$ quorum slices]
  \label{def:qset_function}
  For a node $v \in \V$ and a quorum set $D \in \QSetSpace$,
  $\qset(v, D)$ maps to the set of all valid quorum slices for $v$ that satisfy $D$,
  i.e.,
  $\qset(v, D): \V \times \: \QSetSpace \to 2^{2^{\V}} := \set{q \subseteq \V \mid v \in q \land q \text{ satisfies } D}$.
\end{definition}

Via the $\qset$ notation, quorum sets and quorum slices
become equivalent representations that can be transformed into one another.
A straightforward (but generally not space-efficient) way to express
\emph{any} $k$ quorum slices $\set{q_i \in 2^{\V} \mid i \in [0, k), v \in q_i}$ of a node $v \in \V$
via a quorum set is
$\qset(v, (\emptyset, \InnerQSets, 1))$,
with $\InnerQSets = \set{(q_i, \emptyset, \len{q_i}) \mid i \in [0, k)}$.
Quorum sets are translated to quorum slices (values of $\Q$) by applying the $\qset$ function.
For example (with $\V = \set{0, 1, 2}$):
\begin{align*}
  \Q(0) &= \qset(0, (\set{1, 2},\emptyset, 1)) = \set{\set{0, 1}, \set{0, 2}, \set{0, 1, 2}}\\
  \Q(1) &= \qset(1, (\set{0, 2},\emptyset, 2)) = \set{\set{0, 1, 2}}\\
  \Q(2) &= \qset(2, (\set{0, 1, 2},\emptyset, 2)) = \set{\set{0, 2}, \set{1, 2}, \set{0, 1, 2}}
\end{align*}
In the above example, $\V = \set{0,1,2}$ and their quorum sets (as per $\Q$) form the FBAS $(\V, \Q)$.
As a way to visualize $(\V, \Q)$,
it can heuristically be represented as a graph where the existence of an edge $(v_i, v_j)$
implies that $v_j$ is included in at least one of $v_i$'s quorum slices:

\begin{center}
  \begin{tikzpicture}
    \node[draw, circle] (v0) at (0, 0) {0};
    \node[draw, circle] (v1) at (-1,  -0.7) {1};
    \node[draw, circle] (v2) at (1, -0.7) {2};

    \path [<->] (v0) edge [draw] (v1);
    \path [<->] (v0) edge [draw] (v2);
    \path [<->] (v1) edge [draw] (v2);
  \end{tikzpicture}
\end{center}

\subsection{Preconditions to liveness}%
\label{sub:quorum}

A consensus system is \emph{live} if it can externalize new values\footnote{
  We content ourselves with a weak notion of liveness whereby a system is live as long as it is
  \emph{non-blocking} \cite{garcia2019deconstructing} for one or more non-faulty nodes,
  i.e., as long as an execution path \emph{exists} that allows one or more non-faulty nodes to make progress.
  This can also be called \emph{plausible liveness}.
}.
A consensus system built upon an FBAS is live if the FBAS contains an intact \emph{quorum}---%
a group of FBAS nodes that can
externalize new values by itself.

\begin{definition}[quorum~\cite{mazieres2015stellar}]
  \label{def:quorum}
  A set of nodes $U \subseteq \V$ in FBAS $(\V, \Q)$
  is a quorum iff $U \neq \emptyset$ and
  $U$ contains a quorum slice for each member---i.e.,
  $\forall v \in U \; \exists q \in Q(v): q \subseteq U$.
\end{definition}

This is equivalent to stating that $U$ satisfies the quorum sets of all $v \in U$.
Quorums are therefore determined by the sum of all individual quorum set configurations.
Continuing the previous example with nodes $\V = \set{0, 1, 2}$, we get the quorums $\Quorums = \set{\set{0,2},\set{0,1,2}}$.
We capture part of the semantics behind quorums by defining what it means for a consensus protocol to \emph{honor} a given FBAS%
---namely that whenever values are externalized for a slot, at least one quorum of nodes must eventually externalize values as well.

\begin{definition}[protocol that honors an FBAS]
  \label{def:honors}
  Let $(\V, \Q)$ be an FBAS such that $\V$ contains only non-faulty nodes,
  $P$ a consensus protocol,
  and $N_i \subseteq \V$ the set of all nodes that, following $P$,
  eventually externalize a value for a given slot $i$.
  We say that
  $P$ \emph{honors} $(\V, \Q)$
  iff any nonempty $N_i$ contains a quorum,
  i.e.,
  $\forall i: N_i = \emptyset \lor \exists U \subseteq N$ such that $U$ is a quorum for $(\V, \Q)$.
\end{definition}

We say that $(\V, \Q)$ has \emph{quorum availability despite faulty nodes}
iff there exists a $U \subseteq \V$ that is a quorum in $(\V, \Q)$ and consists of only non-faulty nodes.
Quorum availability despite faulty nodes is a necessary condition to achieving \emph{liveness} in an FBAS,
i.e., ensuring that non-faulty nodes can externalize new values independently of the behavior of faulty nodes
\cite{mazieres2015stellar}.

\begin{theorem}[quorum availability $\Longleftarrow$ liveness]
  \label{theorem:liveness}
  Let $(\V, \Q)$ be an FBAS and $P$ a consensus protocol that honors $(\V, \Q)$.
  If $P$ can provide liveness for $(\V,\Q)$ independently of the behavior of faulty nodes,
  then $(\V, \Q)$ enjoys quorum availability despite faulty nodes.
\end{theorem}
\begin{proof}
  Let $F \subseteq \V$ be the set of all faulty nodes
  and $(\V \setminus F, \Q^\prime)$
  a sub-FBAS that contains all non-faulty nodes,
  with
  $\Q^\prime(v) := \set{q \in \Q(v) \mid q \subseteq \V \setminus F}$
  for $\forall v \in \V \setminus F$.
  $P$ honors $(\V, \Q)$ and can provide liveness independently of the behavior of nodes in $F$,
  therefore there must exist a protocol $P^\prime$
  that can provide liveness while honoring $(\V \setminus F, \Q^\prime)$.
  Based on \cref{def:honors},
  there is therefore at least one $U \subseteq \V \setminus F$
  that is a quorum for $(\V \setminus F, \Q^\prime)$.
  $U$ is, trivially, also a quorum for $(\V, \Q)$.
  \qed
\end{proof}

Given quorum availability despite faulty nodes,
protocols like SCP can provide liveness~\cite{mazieres2015stellar}.
In the case of SCP,
this was previously demonstrated through correctness proofs \cite{garcia2019deconstructing}
as well as formal verification and practical deployment experience \cite{lokhava2019stellar_payments}.
Additional conditions to achieving liveness include the reaction
(via quorum set adaptations, i.e., changes to $\Q$)
to (detectable) timing attacks \cite{lokhava2019stellar_payments}.
We defer to works such as \cite{mazieres2015stellar,losa2019stellar_instantiation,cachin2019asymmetric,cachin2020asymmetric}
for an in-depth exploration of the mechanics and guarantees of consensus protocols for the FBAS setting.

\subsection{Preconditions to safety}%
\label{sub:intersection}

A set of nodes in an FBAS enjoy \emph{safety} if no two of them ever externalize different values for the same slot~\cite{mazieres2015stellar}.
In a blockchain context, a lack of safety guarantees translates into the possibility of forks and double spends.
Protocols that honor an FBAS can only guarantee safety if the FBAS enjoys \emph{quorum intersection}.

\begin{definition}[quorum intersection~\cite{mazieres2015stellar}]
  A given FBAS enjoys quorum intersection iff any two of
  its quorums share a node---i.e., for all quorums
  $U_{1}$ and $U_{2}$, $U_{1} \cap U_{2} \neq \emptyset$.
\end{definition}

For example, the set of quorums $\set{\set{0,2},\set{0,1,2}}$ intersects, whereas introducing an
additional quorum $\set{1,4}$ would break quorum intersection.
In the latter scenario, $\set{0,2}$ and $\set{1,4}$ could induce two new, separated FBASs~\cite{losa2019stellar_instantiation}.
We say that an FBAS enjoys
\emph{quorum intersection despite faulty nodes}
if every two quorums that contain non-faulty nodes intersect in at least one non-faulty node,
even if all faulty nodes change their quorum sets in arbitrary ways or report different quorum sets to different peers.
Formally,
quorum intersection despite faulty nodes is defined via a \emph{delete} operation that
transforms an FBAS based on the assumption that a given set of nodes is acting in the most harmful (to safety) way possible.

\begin{definition}[delete~\cite{mazieres2015stellar}]
  \label{def:delete}
  If $(\V,\Q)$ is an FBAS and $F \subseteq \V$
  a set of nodes, then to \emph{delete} $F$ from $(\V,\Q)$,
  written $(\V,\Q)^F$, means to compute the modified FBAS $(\V \setminus F, \Q^F)$ where
  $\Q^F(v) = \set{q \setminus F, q \in \Q(v)}$.
\end{definition}

If $F \subseteq \V$ is the set of all faulty nodes,
then an FBAS $(\V, \Q)$ enjoys quorums intersection despite faulty nodes iff $(\V, \Q)^F$ enjoys quorum intersection.
If quorum intersection despite faulty nodes is not given, safety cannot be guaranteed (although it can be maintained by chance).

\begin{theorem}[quorum intersection $\Longleftarrow$ guaranteed safety]
  \label{theorem:safety}
  Let $(\V, \Q)$ be an FBAS and
  $P$ a consensus protocol that can provide liveness for any FBAS with quorum availability despite faulty nodes,
  while honoring the respective FBAS.
  Let $P$ furthermore be non-trivial,
  in the sense that externalized values are non-deterministic and depend on user input.
  If $P$ can guarantee safety for all non-faulty nodes in $\V$,
  then $(\V, \Q)$ enjoys quorum intersection despite faulty nodes.
\end{theorem}
\begin{proof}
  Let $F \subseteq \V$ be the set of all faulty nodes and $(\V^\prime, \Q^\prime) := (\V, \Q)^F$.
  If $(\V^\prime, \Q^\prime)$ does not enjoy quorum intersection,
  then there are two quorums $U_1, U_2 \subset \V^\prime$ so that $U_1 \cap U_2 = \emptyset$.
  For $i \in \set{1, 2}$, let $Q_i$ be defined such that $\forall v \in U_i: Q_i(v) := \set{q \in \Q^\prime(v) \mid q \subseteq U_i}$.
  Then both $(U_1, Q_1)$ and $(U_2, Q_2)$ form FBASs with quorum availability.
  As $P$ can provide liveness for any FBAS with quorum availability,
  $(U_1, Q_1)$ and $(U_2, Q_2)$ can externalize values for the same slots
  without any communication taking place between nodes in $U_1$ and nodes in $U_2$.
  As $P$ is non-trivial, the externalized values can differ,
  i.e., safety cannot be guaranteed.
  \qed
\end{proof}

As formally proven by Garc{\'i}a-P{\'e}rez and Gotsman \cite{garcia2018fbqs},
an FBAS that enjoys quorum intersection induces a Byzantine quorum system~\cite{malkhi1998byzantine},
and an FBAS that enjoys quorum intersection despite faulty nodes can induce a dissemination quorum system~\cite{malkhi1998byzantine}.
These results are independent of attempts by faulty nodes to lie about their quorum set configuration \cite{garcia2018fbqs}.
There is strong evidence that protocols like SCP can guarantee safety in any FBAS with quorum
intersection despite faulty nodes~\cite{garcia2019deconstructing,lokhava2019stellar_payments,losa2019stellar_instantiation,cachin2019asymmetric}.

\section{Concepts for further analysis}%
\label{sec:analysis}

In the following, we define new concepts for capturing relevant properties of concrete FBAS instances.
While it is typical in the BFT literature to construct proofs based on \emph{assuming} which sets of nodes \emph{can} fail simultaneously
(i.e., which are the \emph{fail-prone sets}~\cite{malkhi1998byzantine}),
we instead investigate which sets of nodes \emph{have to} fail in order for global liveness and safety guarantees to become void.
This perspective uncovers the liveness and safety \emph{buffers} a given (potentially non-trivial)
quorum system has and is thus highly relevant for the monitoring and evaluation of systems deployed in practice.
While defined based on the FBAS model,
the proposed concepts are readily transferable to more general quorum system formalizations
(e.g., recall that safety-enabling FBASs induce Byzantine quorum systems \cite{garcia2018fbqs}).

For illustration, we will be using the example FBAS defined via \cref{fig:5_node_example}.
An analysis of a slightly larger example FBAS is presented in \cref{app:cascading_example}.
\Cref{app:proofs} contains formal write-ups and proofs of various corollaries and theorems relevant to this section.

\begin{figure}[ht]
  \centering
  \fbox{
    \begin{minipage}{0.3\textwidth}
      $\V = \set{0,1,2,3,4}$\\
      $\Q(0) = \qset(0, (\set{0, 1, 2, 3, 4}, \emptyset, 3))$\\
      $\Q(1) = \qset(1, (\set{0, 1, 2}, \emptyset, 3))$\\
      $\Q(2) = \qset(2, (\set{0, 1, 2}, \emptyset, 3))$\\
      $\Q(3) = \qset(3, (\set{0, 3, 4}, \emptyset, 3))$\\
      $\Q(4) = \qset(4, (\set{0, 3, 4}, \emptyset, 3))$
    \end{minipage}
    \begin{minipage}{0.15\textwidth}
      \begin{center}
        \begin{tikzpicture}
          \node[draw, circle] (v0) at (0, 0) {0};
          \node[draw, circle] (v1) at (-1,  0.5) {1};
          \node[draw, circle] (v2) at (-1, -0.5) {2};
          \node[draw, circle] (v3) at ( 1,  0.5) {3};
          \node[draw, circle] (v4) at ( 1, -0.5) {4};

          \path [<->] (v0) edge [draw] (v1);
          \path [<->] (v0) edge [draw] (v2);
          \path [<->] (v0) edge [draw] (v3);
          \path [<->] (v0) edge [draw] (v4);
          \path [<->] (v1) edge [draw] (v2);
          \path [<->] (v3) edge [draw] (v4);
        \end{tikzpicture}\\
        (heuristic graph representation)
      \end{center}
    \end{minipage}
  }
  \caption{Example FBAS $(\V, \Q)$}
  \label{fig:5_node_example}
\end{figure}

\subsection{Starting point: Minimal quorums}%
\label{sub:minimal_quorums}

As a prerequisite to subsequent analyses, it is helpful to understand which \emph{quorums} (cf. \cref{def:quorum}) exist in an FBAS.
We will be focusing on \emph{minimal quorums},
i.e., quorums $\Min{U} \subseteq \V$ for which there is no proper subset $U \subset \Min{U}$ that is also a quorum.
Informally,
the set of all minimal quorums $\Min{\Quorums}$ carries sufficient information for precisely determining FBAS-wide liveness properties,
while being of significantly smaller size than the set of \emph{all} quorums $\Quorums$.

\begin{definition}[minimal node set]
  \label{def:minimal_set}
  Within the set of node sets $\NodeSetSet \subseteq 2^{\V}$,
  a member set $\Min{N} \in \NodeSetSet$ is minimal iff
  none of its proper subsets is included in $\NodeSetSet$---i.e.,
  $\forall N \in \NodeSetSet, N \not\subset \Min{N}$.
\end{definition}

The FBAS depicted in \cref{fig:5_node_example} has the quorums $\Quorums = \set{\set{0,1,2}, \set{0,3,4}, \set{0,1,2,3,4}}$
and consequently the minimal quorums $\Min{\Quorums} = \set{\set{0,1,2}, \set{0,3,4}}$.

The notion of minimal quorums is helpful, among other things,
for efficiently determining whether an FBAS enjoys quorum intersection~\cite{lachowski2019complexity}:
it can be shown that an FBAS enjoys quorum intersection iff every two of its minimal quorums intersect (\cref{cor:min_quorum_intersection}).

\subsection{Minimal blocking sets}%
\label{sub:liveness}

As per \cref{theorem:liveness},
an FBAS $(\V, \Q)$ cannot enjoy liveness if it doesn't contain at least one non-faulty quorum.
Considering the state of the art in consensus protocols for the FBAS setting and their formal verification (s.a. \cref{sub:quorum}),
quorum availability despite faulty nodes is furthermore the \emph{only} precondition to achieving liveness that depends on $(\V, \Q)$
and arguably the most difficult to satisfy in a practical deployment.
However, while quorum availability can easily be checked based on $\Q$,
faulty nodes are usually not readily identifiable as such in practice.
We therefore propose,
as a means to grasping liveness risks,
to look at sets of nodes that, if faulty, can undermine quorum availability.

\begin{definition}[blocking set]
  \label{def:blocking_set}
  Let $\Quorums \subseteq 2^{\V}$ be the set of all quorums of the FBAS $(\V, \Q)$.
  We denote the set $B \subseteq \V$ as \emph{blocking} iff it intersects every quorum of the FBAS---i.e.,
  $\forall U \in \Quorums, B \cap U \neq \emptyset$
\end{definition}

For example: $\set{0}$ and $\set{1,3}$ are both blocking sets for
$\Quorums = \set{\set{0,1,2}, \set{0,3,4}, \set{0,1,2,3,4}}$.

\begin{corollary}[blocking sets and liveness]
  Control over any blocking set $B$ is sufficient for compromising the liveness of an FBAS $(\V, \Q)$.
\end{corollary}
\begin{proof}
  As $B$ intersects all quorums of the FBAS, there is no quorum that can be formed without cooperation by $B$.
  Without at least one non-faulty quorum, liveness is not possible as per \cref{theorem:liveness}.
  \qed
\end{proof}

Notably, blocking sets can also block liveness \emph{selectively},
enabling \emph{censorship}.
As nodes from the blocking set are present in every quorum,
consensus will never be reached on any value that the blocking set opposes to.
For example, in the context of Stellar,
the blocking set could block the ratification of transactions involving specific accounts.
We chose the term blocking in analogy to the \emph{v-blocking sets} introduced in \cite{mazieres2015stellar}.
As an important distinction,
we use the term \emph{blocking set} to refer to a property of the whole FBAS $(\V, \Q)$,
as opposed to a property of an individual node $v \in \V$.

In the above example, $\set{0}$ and $\set{1,3}$ are not only blocking sets with respect to
$\Quorums$,
they are \emph{minimal blocking sets}, i.e., none of their proper subsets is a blocking set\footnote{
  For completeness, the set of all minimal blocking sets w.r.t. $\Quorums$ is
  $\Min{\BlockSetSet} = \set{\set{0}, \set{1,3}, \set{1,4}, \set{2,3}, \set{2,4}}$.
}.
In essence, minimal blocking sets describe minimal threat (respectively, fail) scenarios w.r.t. liveness.

\subsection{Minimal splitting sets}%
\label{sub:safety}

As per \cref{theorem:safety},
an FBAS can only be considered safe (as one coherent system) as long as it enjoys quorum intersection despite faulty nodes,
i.e., as long as each two of its quorums intersect even after all faulty nodes have been deleted (as per \cref{def:delete}).
For practical purposes,
quorum intersection despite faulty nodes is furthermore a \emph{sufficient} condition for achieving safety in an FBAS,
considering protocols like SCP and the correctness proofs surrounding them (s.a. \cref{sub:intersection}).
Hence, for assessing the risk to safety,
it is interesting to identify sets of nodes that can cause an FBAS to effectively lose quorum intersection.
We call such a set of nodes a \emph{splitting set}, as it can, if faulty, cause at least two quorums to diverge, splitting the FBAS.

\begin{definition}[splitting set]
  \label{def:splitting_set}
  We denote the set $S \subseteq \V$ a splitting set
  iff $(\V, \Q)^S$ lacks quorum intersection---i.e.,
  there are distinct quorums $U_{1}$ and $U_{2}$ of $(\V, \Q)^S$ so that $U_{1} \cap U_{2} = \emptyset$.
\end{definition}

In the above example with $\Min{\Quorums} = \set{\set{0,1,2},\set{0,3,4}}$,
$\set{0}$ is already a splitting set,
as $(\V, \Q)^{\set{0}}$ induces the two non-intersecting quorums $\set{1,2}$ and $\set{3,4}$.
Intuitively, ${\set{0}}$ is a splitting set of $(\V, \Q)$ because it forms the intersection of the quorums $\set{0,1,2}$ and $\set{0,3,4}$.

The existence of a faulty splitting set violates quorum intersection despite faulty nodes
and therefore, as per \cref{theorem:safety}, threatens safety.
Informally,
the members of
a splitting set can perform two types of actions to compromise safety in practice
(s.a. \cref{theorem:minimal_splitting_set_nodes_are_quorum_expanders_or_top_tier}).
On the one hand,
they can change their quorum configurations (or lie about them) to cause existing quorums to shrink
or new quorums to emerge,
both with the goal of reducing the overlap between quorums.
On the other hand,
whenever the intersection of two
(minimal)
quorums is comprised entirely of faulty nodes,
these nodes can agree to different statements in each quorum,
causing the quorums to externalize conflicting values and in this way diverge.

As with blocking sets, we are especially interested in finding the \emph{minimal splitting sets}
$\Min{\SplitSetSet} \subset 2^{\V}$
of an FBAS%
\footnote{
  In the above example, $\set{0}$ is the only minimal splitting set w.r.t. $\Quorums$, i.e., the set of all minimal splitting sets is $\Min{\SplitSetSet} = \set{\set{0}}$.
}
$(\V, \Q)$.
Minimal splitting sets describe minimal threat scenarios w.r.t. safety.

\subsection{Top tier}%
\label{sub:top_tier}

For narrowing down notions of \enquote{centralization} with respect to FBASs, we propose the concept of a \emph{top tier}.
Informally,
the top tier is the set of nodes in the FBAS that is exclusively relevant when determining minimal blocking sets
and hence the liveness buffers
of an FBAS.

\begin{definition}[top tier]
  \label{def:top_tier}
  The top tier of an FBAS $(\V, \Q)$ is the set of all nodes that are contained in one or more minimal quorums---i.e.,
  if $\Min{\Quorums} \subseteq 2^{\V}$ is the set of all minimal quorums of the FBAS,
  $T=\bigcup{\Min{\Quorums}}$ is its top tier.
\end{definition}

In the above example, it in fact holds that $T = \set{0,1,2,3,4} = \V$.

It can be shown that each minimal blocking set consists exclusively of top tier nodes (\cref{cor:min_block_from_top_tier}),
and each top tier node is included in at least one minimal blocking set (\cref{theorem:top_tier_from_min_block}).
The FBAS $(\V, \Q)$ with top tier $T$ has therefore the same properties w.r.t. global liveness as the FBAS induced by $T$,
i.e., the FBAS $(T, \Q^\prime)$ with $\Q^\prime(v) := \set{q \cap T \mid q \in \Q(v)}$.

This observation has direct implications for the computational complexity of FBAS analysis
(further discussed in \cref{sec:algorithms}),
\emph{and} for the performance of FBAS-based consensus protocols.
A consensus round in SCP
(the so far only production-ready protocol for the FBAS setting, to the best of our knowledge)
can demonstrably be completed in $O(\len{T}^2)$ messages.
While classical consensus protocols with quadratic message complexity (such as PBFT~\cite{castro1999practical})
are notorious for becoming unusable in larger validator groups,
several improved protocols have recently emerged that target the blockchain use case and scenarios with 100 and more validators
\cite{yin2019hotstuff_podc,stathakopoulou2019mirbft}.
As a possible avenue for future exploration---%
for FBASs with a \emph{symmetric top tier},
existing permissioned protocols could be adapted without much modification.

\begin{definition}[symmetric top tier]
  \label{def:symmetric_top_tier}
  The top tier $T$ of an FBAS $(\V, \Q)$ is a symmetric top tier iff all top tier nodes have identical quorum sets---i.e.,
  $\exists D \in \QSetSpace, \forall v \in T: \Q(v) = \qset(v, D)$.
\end{definition}

Symmetric top tiers are also significantly more amenable to analysis.
For example, in FBASs with a symmetric top tier $T$ and a non-nested top tier quorum set $(T, \emptyset, t)$,
it holds that any minimal blocking set has cardinality $\len{\Min{B}} = \len{T}-t+1$
(\cref{theorem:symmetric_top_tier_blocking_cardinalities})
and any
minimal splitting set that can cause two top tier nodes to diverge from each other
has cardinality $\len{\Min{S}} = 2t-\len{T}$ (\cref{theorem:symmetric_top_tier_splitting_cardinalities}).

\section{Analysis algorithms}%
\label{sec:algorithms}

In the following,
we propose algorithms for performing the analyses introduced in \cref{sec:analysis}.
We describe them as pseudocode that necessarily abstracts away some implementation details and optimizations.
As a companion to this paper, we release a well-tested implementation of the presented algorithms as open source
(\texttt{fbas\_analyzer}\footnote{
  \url{https://github.com/wiberlin/fbas_analyzer};
  Our Rust-based library has been integrated into \url{https://stellarbeat.io/}
  (a popular monitoring service for the Stellar network)
  and supports in-browser usage---cf. our interactive analysis website at
  \url{https://trudi.weizenbaum-institut.de/stellar_analysis/}.
}).
After outlining algorithms for enumerating minimal quorums (foundation for further analyses),
determining quorum intersection (necessary condition for safety),
enumerating minimal blocking sets (liveness \enquote{buffers}),
enumerating minimal splitting sets (safety \enquote{buffers}),
and efficiently dealing with symmetric top tiers,
the section concludes with a short empirical study on analysis scalability.

\subsection{Minimal quorums}%
\label{sub:minimal_quorums_algo}

\Cref{algo:find_minimal_quorums} describes a branch-and-bound algorithm for finding all minimal quorums.
It is based on a quorum enumeration procedure originally described in \cite{lachowski2019complexity}.
Previous algorithms did not rigorously filter out non-minimal quorums,
which we realize through \texttt{is\_minimal\_quorum}.
The set of all minimal quorums of an FBAS defines its top tier (cf. \cref{sub:top_tier}) and can be used for determining whether the FBAS enjoys quorum intersection.

\begin{algorithm}
  \SetKwProg{Fn}{Function}{:}{}
  \SetKwFunction{FindMinimalQuorums}{find\_minimal\_quorums}
  \SetKwFunction{FindMinimalQuorumsStep}{fmq\_step}
  \SetKwFunction{IsQuorum}{is\_quorum}
  \SetKwFunction{IsMinimalQuorum}{is\_minimal\_for\_quorum}
  \SetKwFunction{IsSatisfiable}{is\_satisfiable}
  \SetKwFunction{ContainsQuorum}{contains\_quorum}
  \Fn{\FindMinimalQuorums{$(\V, \Q)$}}{
    \KwData{An FBAS $(\V, \Q)$.}
    \KwResult{$\Min{\Quorums}$, the set of all minimal quorums of $(\V, \Q)$.}
    \BlankLine
    $V \leftarrow \V$ sorted by, e.g., PageRank~\cite{page1999pagerank} (cf. \cite{lachowski2019complexity})\;
    \KwRet{\FindMinimalQuorumsStep{$\emptyset$, $V$, $\Q$}}\;
  }
  \BlankLine
  \Fn{\FindMinimalQuorumsStep{$U$, $V$, $\Q$}}{
    \uIf{\IsQuorum{$U$, $\Q$}}{
      \uIf{\IsMinimalQuorum{$U$, $\Q$}}{
        \KwRet{$\set{U}$}\;
      }
      \lElse{\KwRet{$\emptyset$}}
    }
    \uElseIf{\IsSatisfiable{$U$, $V$, $\Q$}}{
      $v \leftarrow$ next in $V$\;
      \KwRet{\FindMinimalQuorumsStep{$U \cup \set{v}$, $V \setminus \set{v}$, $\Q$}
        \ $\cup$ \
      \FindMinimalQuorumsStep{$U$, $V \setminus \set{v}$, $\Q$}}\;
    }
    \lElse{\KwRet{$\emptyset$}}
  }
  \BlankLine
  \Fn{\IsQuorum{$U$, $\Q$}}{
    \KwRet{$\forall v \in U \; \exists q \in \Q(v): q \subseteq U$}\;
  }
  \BlankLine
  \Fn{\IsSatisfiable{$U$, $V$, $\Q$}}{
    \KwRet{$\forall v \in U \; \exists q \in \Q(v): q \subseteq U \cup V$}\;
  }
  \BlankLine
  \Fn{\IsMinimalQuorum{$U$, $\Q$}}{
    \For{$v \in U$} {
      \uIf{\ContainsQuorum{$U \setminus \set{v}$, $\Q$}}{
        \KwRet{false}\;
      }
    }
    \KwRet{true}\;
  }
  \BlankLine
  \Fn{\ContainsQuorum{$U$, $\Q$}}{
    \tcp{remove non-satisfiable nodes}
    \While{$\exists v \in U \; \forall q \in \Q(v): q \not\subseteq U$} {
      $U \leftarrow \set{v \in U \mid \exists q \in \Q(v): q \subseteq U}$\;
    }
    \KwRet{$U \neq \emptyset$}\;
  }
  \caption{Find minimal quorums}
  \label{algo:find_minimal_quorums}
\end{algorithm}

The keystone of the algorithm is the function \texttt{fmq\_step} that takes a current quorum \emph{candidate} $U$,
a sorted list of \emph{yet-to-be-considered nodes} $V$ and a reference to $\Q$ for mapping nodes to their quorum sets.
The algorithm implements a classical branching pattern:
at each invocation of \texttt{fmq\_step} in which $U$ is not already a quorum,
the next node in $V$ is taken out and, in one branch, added to $U$, and, in the other, not.
Hopeless branches are identified early using the $\texttt{is\_satisfiable}$ function.

As proposed in \cite{lachowski2019complexity},
we initially sort $V$ using a heuristic such as PageRank~\cite{page1999pagerank}
which can improve the algorithm's performance in practice.
Another important optimization from \cite{lachowski2019complexity},
that we leave out in our pseudocode for greater clarity,
is the partitioning of $\V$ into \emph{strongly connected components}\footnote{
  Based on the heuristic representation of the FBAS as a directed graph.
} so that \texttt{find\_minimal\_quorums} must be applied only to (often significantly smaller) subsets of $\V$.
Tarjan~\cite{tarjan1972depth} gives an algorithm for performing this preprocessing step in linear time.

As noted in other works (e.g., \cite{lachowski2019complexity,bracciali2021decentralization}),
determining quorum intersection, and hence also enumerating all minimal quorums, is NP-hard.
Consequently, our algorithm has exponential time complexity.
For an FBAS with $n = \len{\V}$ nodes and a top tier of size $\tts = \len{T}$ we find all $k \leq \binom{\tts}{\ceil{\frac{\tts}{2}}}$ minimal quorums in $O(2^n)$.
Note that in practice the number of de-facto considered nodes $n$ is greatly reduced through polynomial-time preprocessing steps such as strongly-connected-component analysis and heuristics-based sorting,
yielding actual running times that are close to the $O(2^\tts)$ bound.

\subsection{Quorum intersection}%
\label{sub:has_quorum_intersection_algo}

Quorum intersection is a central property for being able to guarantee safety in an FBAS (cf. \cref{sub:safety}).
Quorum intersection can be determined by checking the pairwise intersection of all minimal quorums (\cref{cor:min_quorum_intersection}).
This straightforward approach, that was also proposed in \cite{lachowski2019complexity},
is embodied in \cref{algo:has_quorum_intersection}.

\begin{algorithm}
  \SetKwProg{Fn}{Function}{:}{}
  \SetKwFunction{HasQuorumIntersection}{has\_quorum\_intersection}
  \SetKwFunction{FindMinimalQuorums}{find\_minimal\_quorums}
  \Fn{\HasQuorumIntersection{$(\V, \Q)$}}{
    \KwData{An FBAS $(\V, \Q)$.}
    \KwResult{\emph{true} if the FBAS enjoys quorum intersection, \emph{false} else.}
    \BlankLine
    $\Min{\Quorums} \leftarrow $ \FindMinimalQuorums{$(\V, \Q)$}\;
    \KwRet{$\forall \Min{U}_i, \Min{U}_j \in \Min{\Quorums}: \Min{U}_i \cap \Min{U}_j \neq \emptyset$}\;
  }
  \caption{Checking for quorum intersection via approach from \cite{lachowski2019complexity}.}
  \label{algo:has_quorum_intersection}
\end{algorithm}

In this paper, we propose an additional, alternative algorithm (\cref{algo:alt_has_quorum_intersection}),
that doesn't check for pairwise intersections but instead checks whether the \emph{complement sets} of found quorums contain quorums themselves.
If this is never the case, the FBAS enjoys quorum intersection.
This approach for checking for quorum intersection has the benefit that only a constant number of node sets must be held in memory at the same time, as opposed to all minimal quorum sets as in \cref{algo:has_quorum_intersection}.
The space complexity of the check is therefore reduced from exponential to linear.

\begin{algorithm}
  \SetKwProg{Fn}{Function}{:}{}
  \SetKwFunction{HasQuorumIntersection}{has\_quorum\_intersection}
  \SetKwFunction{FindMinimalQuorums}{find\_minimal\_quorums}
  \SetKwFunction{ContainsQuorum}{contains\_quorum}
  \Fn{\HasQuorumIntersection{$(\V, \Q)$}}{
    \KwData{An FBAS $(\V, \Q)$.}
    \KwResult{\emph{true} if the FBAS enjoys quorum intersection, \emph{false} else.}
    \BlankLine
    \For{$\Min{U} \in$ \FindMinimalQuorums{$(\V, \Q)$}}{
      \uIf{\ContainsQuorum{$\V \setminus \Min{U}$}}{
        \KwRet{false}\;
      }
    }
    \KwRet{true}\;
  }
  \caption{Checking for quorum intersection via alternative approach with linear space complexity.}
  \label{algo:alt_has_quorum_intersection}
\end{algorithm}

Our implementation of \cref{algo:alt_has_quorum_intersection} is also empirically faster for many FBASs,
probably because \texttt{contains\_quorum} scales better than iterating once over all minimal quorums,
and because less data must be written to memory.
For both algorithms, we leave out optimization details
such as leveraging the fact that quorum intersection is guaranteed to hold if all minimal quorums $\Min{U} \in \Min{\Quorums}$ have cardinality greater than $\frac{\len{\bigcup \Min{\Quorums}}}{2}$.
In \cref{algo:alt_has_quorum_intersection},
for example,
it suffices to check only minimal quorums with fewer than $\frac{\len{\bigcup \Min{\Quorums}}}{2}$ members.

\subsection{Minimal blocking sets}%
\label{sub:minimal_blocking_sets_algo}

\Cref{algo:find_minimal_blocking_sets} presents our algorithm for enumerating all minimal blocking sets based on a branch-and-bound strategy.
The check whether a given \emph{candidate set} $B$ is blocking is performed by checking whether the FBAS contains any quorums after $B$ is removed from the node population.
If a blocking set can still be formed from $B$ and the yet-to-be-considered nodes $V$ (this is the pruning rule),
the enumeration continues, branching via either adding the next node in $V$ to the candidate set or discarding it altogether.
The order in which nodes are visited can be tuned using a suitable heuristic---we sort nodes using PageRank~\cite{page1999pagerank}
(as for finding minimal quorums) in the example pseudocode and our current implementation.
Like for \cref{algo:find_minimal_quorums},
the complexity of \cref{algo:find_minimal_blocking_sets} is in $O(2^n)$
(for an FBAS with $n$ nodes)
with a likely practical average case complexity of $O(2^\tts)$
($\tts$ being the size of the top tier).

\begin{algorithm}
  \SetKwProg{Fn}{Function}{:}{}
  \SetKwFunction{FindMinimalBlockingSets}{find\_minimal\_blocking\_sets}
  \SetKwFunction{FindMinimalBlockingSetsStep}{fmb\_step}
  \SetKwFunction{IsMinimalBlocking}{is\_minimal\_for\_blocking}
  \SetKwFunction{IsBlocking}{is\_blocking}
  \SetKwFunction{ContainsQuorum}{contains\_quorum}
  \Fn{\FindMinimalBlockingSets{$(\V, \Q)$}}{
    \KwData{An FBAS $(\V, \Q)$.}
    \KwResult{$\Min{\BlockSetSet}$, the set of all minimal blocking sets of $(\V, \Q)$.}
    \BlankLine
    $V \leftarrow \V$ sorted by, e.g., PageRank~\cite{page1999pagerank}\;
    \KwRet{\FindMinimalBlockingSetsStep{$\emptyset$, $V$, $\Q$}}\;
  }
  \BlankLine
  \Fn{\FindMinimalBlockingSetsStep{$B$, $V$, $\Q$}}{
    \uIf{\IsBlocking{$B$, $V$, $\Q$}}{
      \uIf{\IsMinimalBlocking{$B$, $V$, $\Q$}}{
        \KwRet{$\set{B}$}\;
      }
      \lElse{\KwRet{$\emptyset$}}
    }
    \uElseIf{\IsBlocking{$B \cup V$, $V$, $\Q$}}{
      $v \leftarrow$ next in $V$\;
      \KwRet{\FindMinimalBlockingSetsStep{$B \cup \set{v}$, $V \setminus \set{v}$, $\Q$}
        \ $\cup$ \
      \FindMinimalBlockingSetsStep{$B$, $V \setminus \set{v}$, $\Q$}}\;
    }
    \lElse{\KwRet{$\emptyset$}}
  }
  \BlankLine
  \Fn{\IsBlocking{$B$, $V$, $\Q$}}{
    \KwRet{$\lnot\text{\ContainsQuorum{$V \setminus B$, $\Q$}}$}\;
  }
  \BlankLine
  \Fn{\IsMinimalBlocking{$B$, $V$, $\Q$}}{
    \For{$v \in B$}{
      \uIf{\IsBlocking{$B \setminus \set{v}$, $V$, $\Q$}}{
        \KwRet{false}\;
      }
    }
    \KwRet{true}\;
  }
  \caption{Find minimal blocking sets}
  \label{algo:find_minimal_blocking_sets}
\end{algorithm}

\subsection{Minimal splitting sets}%
\label{sub:minimal_splitting_sets_algo}

\Cref{algo:find_minimal_splitting_sets} presents our algorithm for enumerating all minimal splitting sets.
We again perform a branch-and-bound search.
The final condition for accepting a \emph{candidate set} $S$ is whether deleting it (cf. \cref{def:delete}) from the FBAS causes the FBAS to lose quorum intersection.

This check is significantly more expensive than the corresponding checks
in \cref{algo:find_minimal_quorums} and \cref{algo:find_minimal_blocking_sets}.
Additionally, unlike the previously presented algorithms,
\cref{algo:find_minimal_splitting_sets} also needs
to consider non-top tier nodes as candidates.
We incorporate the observation (from \cref{theorem:minimal_splitting_set_nodes_are_quorum_expanders_or_top_tier})
that a node can only be part of a minimal splitting set if it is part of a minimal quorum
(only then can it be part of an intersection of minimal quorums)
or if a change of its quorum set can potentially cause new, smaller quorums to emerge.
Consequently,
we consider as candidates
all top tier nodes and all nodes that are \emph{quorum expanders}:
nodes that are part of a quorum slice of another node that is a not a quorum slice for themselves
(formal definition in \cref{def:quorum_expanders}).
Informally, by not sharing a quorum slice with a node they affect,
quorum expanders may force quorums to expand beyond this quorum slice.
By changing their quorum set,
quorum expanders could reverse this effect, leading to smaller quorums and, accordingly,
an increased risk to quorum intersection.

The \texttt{has\_potential} function embodies
an explicit pruning condition for the branch-and-bound search.
Here, we check whether a change in the FBAS's minimal quorums is possible if all outstanding candidate nodes $V$ are joined with the current candidate set $S$.
As a heuristic to avoid actually calculating minimal quorums,
we check whether the quorum-containing strongly connected components of the FBAS change after deleting $V$ in addition to $S$.

For improving readability and comprehension,
we leave out various details and smaller optimizations from our pseudocode listing for \cref{algo:find_minimal_splitting_sets}.
Among other things,
we don't include our full algorithms for enumerating \texttt{quorum\_expanders}
and deliberately ignore opportunities for caching and reusing the results of costly operations.

\begin{algorithm}
  \SetKwProg{Fn}{Function}{:}{}
  \SetKwFunction{HasQuorumIntersection}{has\_quorum\_intersection}
  \SetKwFunction{FindMinimalSplittingSets}{find\_minimal\_splitting\_sets}
  \SetKwFunction{FindSplittingSetsStep}{fs\_step}
  \SetKwFunction{QuorumExpanders}{quorum\_expanders}
  \SetKwFunction{ConsensusClusters}{quorum\_clusters}
  \SetKwFunction{HasPotential}{has\_potential}
  \SetKwFunction{ReduceToMinimalSets}{reduce\_to\_minimal\_sets}
  \Fn{\FindMinimalSplittingSets{$(\V, \Q)$}}{
    \KwData{An FBAS $(\V, \Q)$.}
    \KwResult{$\Min{\SplitSetSet}$, the set of all minimal splitting sets of $(\V, \Q)$.}
    \BlankLine
    $V \leftarrow \bigcup{\text{\FindMinimalQuorums{$(\V, \Q)$}}}$\;
    $V \leftarrow V \cup \QuorumExpanders{$(\V, \Q)$}$\;
    $V \leftarrow V$ sorted by, e.g., number of affected nodes\;
    $A \leftarrow \V$\;
    $\SplitSetSet \leftarrow$ \FindSplittingSetsStep{$\emptyset$, $V$, $A$, $(\V, \Q)$}\;
    \KwRet{\ReduceToMinimalSets{$\SplitSetSet$}}\;
  }
  \BlankLine
  \Fn{\FindSplittingSetsStep{$S$, $V$, $(\V, \Q)$}}{
    \uIf{$\lnot\HasQuorumIntersection{$(\V, \Q)^S$}$}{
      \KwRet{$\set{S}$}\;
    }
    \uElseIf{\HasPotential{$S$, $V$, $(\V, \Q)$}}{
      $v \leftarrow$ next in $V$\;
      \KwRet{\\
        \ \FindSplittingSetsStep{$S \cup \set{v}$, $V \setminus \set{v}$, $(\V, \Q)$}
        $\cup$\\
        \ \FindSplittingSetsStep{$S$, $V \setminus \set{v}$, $(\V, \Q)$}
      }\;
    }
    \lElse{\KwRet{$\emptyset$}}
  }
  \BlankLine
  \Fn{\QuorumExpanders{$(\V, \Q)$}}{
    \KwRet{$\set{v \in \V \mid \exists v^\prime \in \V, q^\prime \in \Q(v^\prime):
    \qquad\qquad\qquad %
    v \in q^\prime \land (\forall q \in \Q(v) : q \not\subseteq q^\prime) }$}\;
  }
  \BlankLine
  \Fn{\HasPotential{$S$, $V$, $(\V, \Q)$}}{
    \KwRet{$\ConsensusClusters{$(\V, \Q)^{S \cup V}$} \neq \ConsensusClusters{$(\V, \Q)^{S}$}$}\;
  }
  \BlankLine
  \Fn{\ConsensusClusters{$(\V, \Q)$}}{
    $\NodeSetSet \leftarrow$ strongly connected components of $(\V, \Q)$\;
    \KwRet{$\set{N \in \NodeSetSet \mid \ContainsQuorum{N}}$}\;
  }
  \BlankLine
  \Fn{\ReduceToMinimalSets{$\SplitSetSet$}}{
    \KwRet{$\set{\Min{S} \in \SplitSetSet \mid \forall S \in \SplitSetSet: S \not\subset \Min{S}}$}\;
  }
  \caption{Find minimal splitting sets}
  \label{algo:find_minimal_splitting_sets}
\end{algorithm}

The asymptotic complexity of \cref{algo:find_minimal_splitting_sets} remains in $O(2^n)$, respectively $O(2^{\len{T \cup X}})$
where $T$ is the top tier and $X$ the set of all quorum expanders.
However, due to the costly acceptance check for splitting sets
and the larger number of nodes that need to be considered,
the algorithm is significantly slower than \cref{algo:find_minimal_quorums} and \cref{algo:find_minimal_blocking_sets} in practice.

\subsection{Symmetric clusters}%
\label{sub:symmetry_algos}

As a generalization of symmetric top tiers (\cref{def:symmetric_top_tier}),
we define \emph{symmetric clusters} of an FBAS $(\V, \Q)$ as groups of nodes
$Y \subseteq \V$ such that $\exists D \in \QSetSpace, \forall v \in Y: \Q(v) = \qset(v, D)$ and $\bigcup{\bigcup{\set{\Q(v), v \in Y}}} = Y$.
If an FBAS has one symmetric cluster $Y$ and $\V \setminus Y$ does not contain a quorum,
$Y$ is the symmetric top tier of $(\V, \Q)$\footnote{
  If an FBAS has $l > 1$ symmetric clusters or $\V \setminus Y$ does contain a quorum, $(\V, \Q)$ does not enjoy quorum intersection.
}.

Symmetric clusters can be found in polynomial time,
by grouping nodes with identical quorum set configurations (values for $\Q$) and checking the above condition for each thus formed candidate set.

Symmetric clusters can be analyzed significantly more efficiently.
For example, an FBAS with a non-nested symmetric top tier is isomorphic to a classical, threshold-based quorum system
(s.a. \cref{theorem:symmetric_top_tier_blocking_cardinalities,theorem:symmetric_top_tier_splitting_cardinalities}).
For symmetric clusters formed around a nested quorum set,
minimal quorums and minimal blocking sets can be enumerated without the overhead of checking candidate sets,
by recursively listing combinations and forming their Cartesian product.
If the interest is to find only such splitting sets that
can cause nodes within the symmetric cluster to diverge,
then the same is true for minimal splitting sets.

\subsection{Analysis performance}%
\label{sub:performance_eval}

Our analysis approach requires the enumeration of minimal quorums, minimal blocking sets and minimal splitting sets---which in all three cases is an NP-hard problem.
It is unclear, however, what this means for the practical limitations of thoroughly determining the safety and liveness buffers of an FBAS.
Practical limitations are difficult to conclusively determine as the real-life performance of analyses depends heavily on the topology of analyzed FBASs and the implementation of the algorithms.

In the following, we present a short exploratory study into the scalability of our own implementation.
We construct synthetic FBASs of increasing size that consist of only a top tier.
In the first series of presented experiments (\cref{fig:performance_ideal}),
we construct FBASs $(\V, \Q)$ resembling classical $3f+1$ quorum systems:
\begin{align*}
  \forall v \in \V: \Q(v) = \qset(v, (\V, \emptyset, \optithresh{\V}))
\end{align*}
In a second series of experiments (\cref{fig:performance_stellarlike}),
we approximate the structure of the Stellar network's top tier where each organization is represented by (usually) 3 physical nodes arranged in crash failure-tolerating $2f+1$ inner quorum sets:
\begin{align*}
  &\V = \set{v_0, v_1, ... v_{n-1}}, n = 3m\\
  &\InnerQSets = \set{(\set{v_{3i}, v_{3i+1}, v_{3i+2}}, \emptyset, 2) \mid i \in [0, m)}\\
  &\forall v \in \V: \Q(v) = \qset(v, (\emptyset, \InnerQSets, \optithreshraw{m}))
\end{align*}

We enumerate all minimal quorums, minimal blocking sets and minimal splitting sets of thus generated FBASs and record the time to completion of each of these operations.
All analyses were single-threaded and performed on regular server-class hardware.
We explicitly deactivated all optimizations based on detecting and exploiting symmetric clusters,
so that the results of this study reflect the performance of the more expensive
\cref{algo:find_minimal_quorums,algo:find_minimal_blocking_sets,algo:find_minimal_splitting_sets}.

\begin{figure}[htpb]
  \centering
  \includegraphics[width=\linewidth]{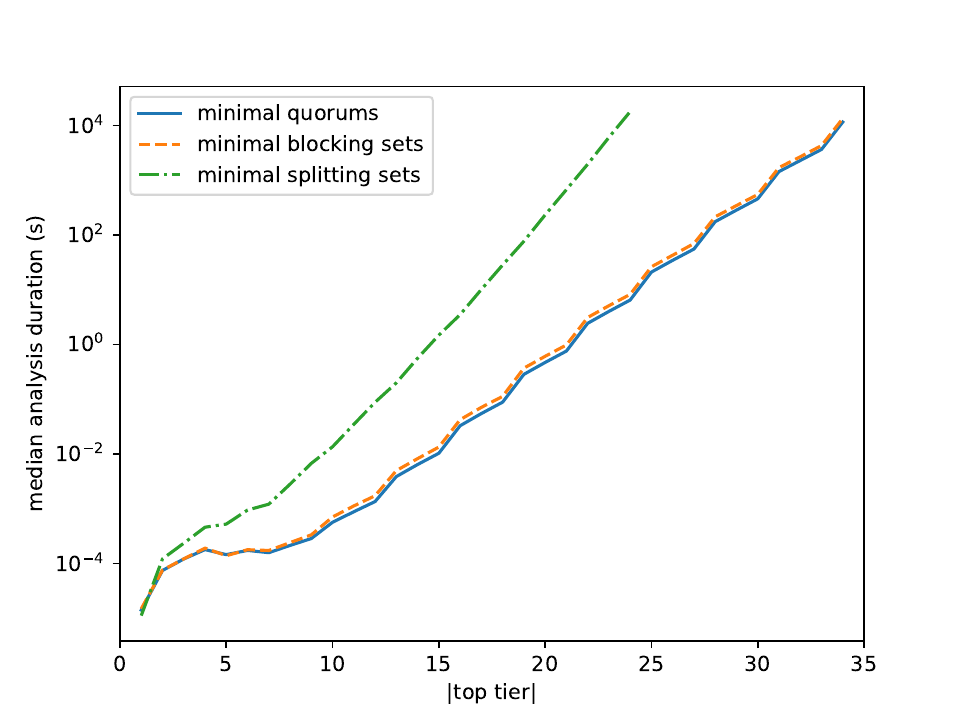}
  \caption{Analysis duration for FBASs resembling classical $3f+1$ quorum systems.
    Analysis optimizations for symmetric top tiers were turned off.}%
  \label{fig:performance_ideal}
\end{figure}

\begin{figure}[htpb]
  \centering
  \includegraphics[width=\linewidth]{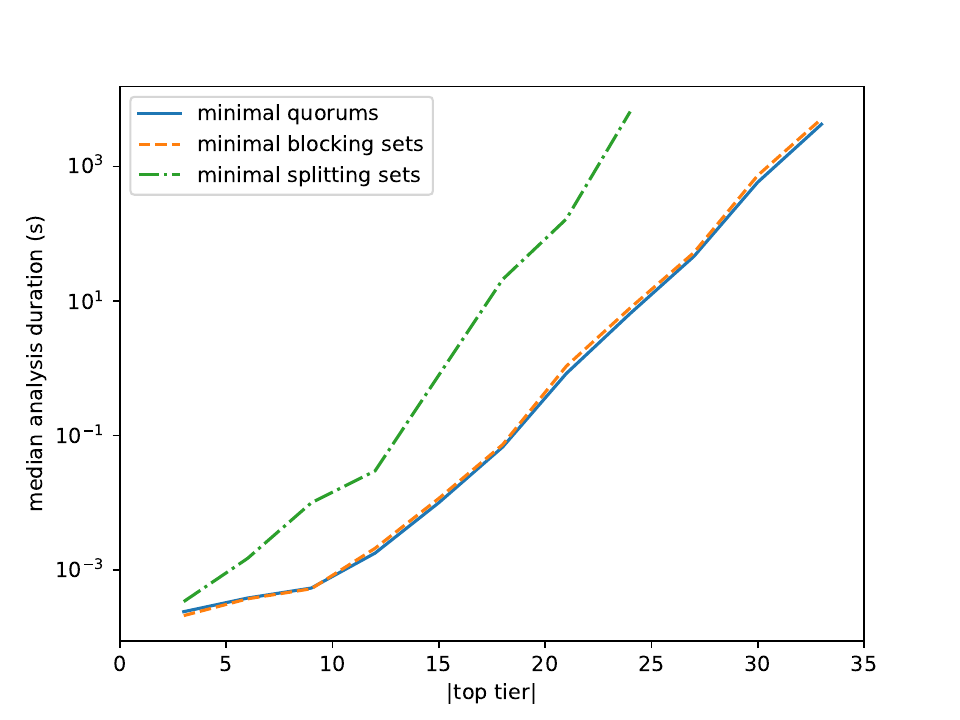}
  \caption{Analysis duration for FBASs resembling the structure of the Stellar network top tier.
    Analysis optimizations for symmetric top tiers were turned off.}%
  \label{fig:performance_stellarlike}
\end{figure}

\Cref{fig:performance_ideal,fig:performance_stellarlike} depict the median measured times on a log scale,
from a set of 10 measurements per FBAS size (we performed the same analysis 10 times, recording individual times).
As was expected, analysis durations raise exponentially with growing top tier sizes $m$.
Analyses start requiring more than an hour to finish at $m \geq 23$ for flat symmetric top tiers and $m \geq 24$ for Stellar-like topologies.
This is a cautiously positive result---top tier sizes observed in practice are currently in the range of 7 organizations (23 raw nodes) for the Stellar network
(cf. \cref{app:stellar_example})
and 7 organizations (10 raw nodes) for the MobileCoin network \cite{ndolo2021crawling}.
It is likely that,
for example through parallelization or the development of additional optimizations for \enquote{almost symmetric} FBASs,
the analysis durations for naturally occurring FBASs can be reduced further.

\section{Bootstrapping FBASs}%
\label{sec:qsc}

The reported openness enabled through the FBAS paradigm comes at the cost of increased configuration responsibilities for node operators.
As discussed in \cref{sec:fba}, each node must become associated with a quorum set (respectively quorum slices) in order to become a useful part of an FBAS.
We will refer to this process as \emph{quorum set configuration} (QSC).
But how should a node operator go about QSC?
Based on the analytical toolset introduced in \cref{sec:analysis},
we can now investigate what kinds of \emph{QSC policies} are plausible and in what kind of FBASs they result.

Notably, we explore how \emph{individual preferences} (such as which nodes should be \enquote{trusted}) can be mapped to the quorum set formalism.
Based on experiments that use Internet topology as a representative graph representation of interdependence and trust,
we conclude that purely individualistic configuration policies can result in systems with low liveness and high complexity.
We outline possible directions for future research by sketching policies with a \emph{strategic} element and empirically demonstrating their effectiveness.

\subsection{QSC policies and their evaluation}%
\label{sub:qsc_policy}

A QSC policy is individually and repeatedly invoked for each node $v \in \V$.
It takes information about a current FBAS instance $(\V, \Q)$ as input and returns a quorum set for $v$,
setting a new value for $\Q(v)$.
We use the quorum set formalization introduced in \cref{sub:quorum_set}.
For illustration, consider the following trivial policy:
\clearpage

\begin{equation}
  \label{qsc:super_safe}
  \tag{Super Safe QSC}
  \forall v \in \V:\quad
  \Q(v) = \qset(v, (\V, \emptyset, \len{\V}))
\end{equation}

If implemented by all nodes in $\V$, \ref{qsc:super_safe} leads to each node having only one quorum slice---$\V$ itself ($\Q(v) = \set{\V}$).
The policy maximizes safety but leads to blocking sets of cardinality 1---any node can block the single quorum in the induced FBAS.

As an improvement, the threshold of the formed quorum sets can be set in resemblance to classical BFT protocols:

\begin{equation}
  \label{qsc:ideal_open}
  \tag{Ideal Open QSC}
  \forall v \in \V:\quad
  \Q(v) = \qset(v, (\V, \emptyset, \optithresh{\V}))
\end{equation}

For $\len{\V} = 3f + 1$ with an $f \in \mathbb{Z}^{+}$, setting the threshold to $t = \optithresh{\V}$
leads to FBASs in which any $2f + 1$ nodes form a (minimal) quorum.
This results in both all minimal blocking sets and all minimal splitting sets of the induced FBAS having cardinality $f + 1$, i.e.,
both safety \emph{and} liveness can be maintained in the face of up to $f$ node failures.

\subsubsection{Choosing validators}%
\label{ssub:choosing_validators}

The preceding example policies construct non-nested quorum sets that use as validators $U$ the set of all nodes in the FBAS ($U = \V$).
These are clearly toy examples---if anything else, without additional mechanisms to restrict or
filter the membership in $\V$, $\V$ can easily become dominated by faulty Sybil~\cite{douceur2002sybil} nodes.

In the scope of this work, and in line with the motivation behind the FBAS paradigm,
we consider $\V$ to enjoy \emph{open membership}, with no universally trusted whitelist or ranking.
For arriving at sensible choices for $U$, QSC policies must therefore take individual knowledge into account.

\subsubsection{Modeling individual preferences}%
\label{ssub:modeling_individual_preference}
QSC policies based on individual preferences contribute node-local knowledge to the collective FBAS configuration.
For example:
\begin{itemize}
  \item Which nodes are \emph{trusted} to be (and stay) non-faulty.
    It is often implied that QSC should reflect some form of trust, e.g.,
    in wordings such as \enquote{flexible trust}~\cite{mazieres2015stellar} or \enquote{asymmetric distributed trust}~\cite{cachin2019asymmetric}.
    While reasoning about the future behavior of participants in a consensus protocol might be an overwhelming task for node operators,
    they may at least encode plausible beliefs about non-Sybilness~\cite{douceur2002sybil}
    (i.e., which groups of nodes are (un)likely to be controlled by the same entity).
  \item To which nodes do \emph{dependencies} exist (e.g., for business reasons).
    Adding nodes of organizations one interacts with to one's quorum sets might be necessary to maintain \enquote{sync} with these organizations~\cite{lokhava2019stellar_payments},
    as opposed to ending up with diverging ledgers in the event of a fork.
\end{itemize}

In the following discussion, we will use \emph{graph representations} for modeling individual preferences.
It is an intriguing hypothesis that the FBAS paradigm can enable
Sybil-resistant and yet energy-efficient permissionless consensus
by bootstrapping quorum systems along
existing trust graphs or interdependence graphs.
In \cref{sub:fbas} we saw that transforming an FBAS into an equally sized regular graph leads to a loss of information,
i.e., can yield only heuristic representations.
In the following sections we pose the inverse question:
How can a \enquote{good} FBAS $(\V, \Q)$ be instantiated from a given graph $G = (\V, E)$?

For evaluating example policies incorporating individual preferences,
we will use the \emph{autonomous system} (AS) relationships graph inferred by the CAIDA project\footnote{
  The CAIDA AS Relationships Dataset,
  1998-01-01 (serial-1) and 2020-01-01 (serial-2),
  \url{https://www.caida.org/data/as-relationships/}
}---a reflection of the interdependence and trust between networks that form the Internet.
The topological structure of the Internet has repeatedly been cited as an argument for the viability of the FBAS model~\cite{mazieres2015stellar,lokhava2019stellar_payments}.
We discuss results based on two snapshots of the AS relations graph:
from January 1998---the earliest available snapshot describing a younger Internet with
$3233$ ASs connected via $4921$ (directed) customer/provider links and $852$ (undirected) peering
links---and from January 2020---with
$67308$ ASs connected via $133864$ customer/provider links and $312763$ peering links.
We will refer to the graphs as $\oldASG$ and $\newASG$.

\subsection{Naive individualistic QSC}%
\label{sub:individual_qsc}

We consider a QSC policy naively individualistic if it is based entirely on \emph{individual preferences}.
We model \enquote{preference for a node} as edges in a graph $G = (\V, E)$, with nodes being aware only of their own graph neighborhood.

Consider a simple representative of this class---forming quorum sets using the entire graph neighborhood of a node,
weighing each neighbor equally within a $3f + 1$ threshold logic
(that models the assumption that strictly less than a third of all neighbors can be faulty):

\begin{equation}
  \label{qsc:all_neighbors}
  \tag{All Neighbors QSC}
  \begin{aligned}
    \forall v \in \V:\quad
        U &= \set{v} \cup \set{v^\prime \in \V \mid (v, v^\prime) \in E}\\
    \Q(v) &= \qset(v, (U, \emptyset, \optithresh{U}))
  \end{aligned}
\end{equation}

If $G$ is a complete graph, we get the same result as with \ref{qsc:ideal_open}.
If $G$ is not connected, we cannot have quorum intersection (and hence safety).
The latter is also true if $G$ contains more than one cluster of sufficient size and weak (relative) connectedness to the rest of the graph.
We can confirm that this is the case for the AS graph snapshots $\oldASG$ and $\newASG$.
Using them, \ref{qsc:all_neighbors} induces FBASs that \emph{do not} enjoy quorum intersection\footnote{
  As determined using \texttt{fbas\_analyzer} (\cref{sec:algorithms}).
}.
The high prevalence of AS peering is a likely explanation for why sufficiently well intraconnected clusters can emerge outside of the \enquote{natural} top tier of the AS graph.

A lack of quorum intersection implies that the induced FBASs may split into multiple sub-FBASs.
This might be a desirable effect when bootstrapping from individual preferences.
For example, separated communities with low levels of inter-community interaction and trust might prefer the added sovereignty of an \enquote{own} FBAS.
We repeated the analysis for the respectively largest sub-FBASs,
with an upper bound on top tier size\footnote{
  Based on the size of the largest quorum that is fully contained in a strongly connected component
  (which is the union of all such quorums).
} of, respectively, $355$ and $14339$ nodes.
Potential top tier sizes of this magnitude make a complete analysis unfeasible (s.a. the discussion on analysis scalability in \cref{sub:performance_eval}).
This is problematic, as the robustness of the resulting FBASs, in terms of safety and liveness, cannot be reliably determined.
Existing weaknesses in the global quorum structure cannot be identified and (strategically) fixed.
Weaknesses, however, are likely to exist.
For example, preliminary analysis results for the FBAS instantiated from $\oldASG$ imply the existence of blocking sets with only $3$ members.

\subsection{Tier-based QSC}%
\label{sub:tier_based_qsc}

Towards making resulting top tiers more focused (and hence, the resulting FBASs more efficient and more amenable to analysis),
QSC policies can incorporate \emph{strategic considerations} in addition to individual preferences.
We explore a prudent example strategy in the following: the weighing of nodes based on \emph{tierness}, or relative importance.
Tierness is an established notion for ASs in the Internet graph.
For FBASs, a tiered quorum structure with every node including only higher-tier neighbors in its quorum sets was proposed (as an example)
as early as in the original FBAS proposal \cite{mazieres2015stellar}.
Classifying nodes based on their tierness is also related to the \emph{quality}-based configuration format currently used by the Stellar software~\cite{lokhava2019stellar_payments}.
Lastly, it is a plausible assumption that the relative tierness of graph neighbors can be estimated locally,
enabling QSC decisions that do not require a global view.

We sketch an example QSC policy in which nodes use only higher-tier nodes in their quorum sets,
or same-tier nodes if none of their neighbor appears to be of higher tier.
We assume that nodes can infer the relative tierness of their graph neighbors.
Specifically, that they can determine which of their neighbors are of a higher tier than themselves.
For simulation, we use the PageRank~\cite{page1999pagerank} score of nodes (calculated without dampening) as a proxy for their tierness.
Each simulated node considers a neighbor of higher (lower) tier if the neighbor's PageRank score is twice as high (low) as its own.
More formally, with $\rank(v)$ denoting the PageRank score of node $v$,
$\outlinks(v)$ the set of its neighbors ($\outlinks(v) := \set{v^\prime \in \V \mid (v, v^\prime) \in E}$),
$H$ its higher-tier neighbors and $P$ its same-tier neighbors (\enquote{peers}):
\begin{equation}
  \label{eq:tier_heuristics}
  \tag{Tierness Heuristics}
  \begin{aligned}
    H(v) = \set{v^\prime \in \outlinks(v) \mid \rank(v^\prime) \geq 2\rank(v)}\\
    P(v) = \set{v^\prime \in \outlinks(v) \mid \frac{1}{2}\rank(v) < \rank(v^\prime) < 2\rank(v)}
  \end{aligned}
\end{equation}

Based on this heuristic, we can define the following QSC policy:
\begin{equation}
  \label{qsc:higher_tier}
  \tag{Higher-Tier Neighbors QSC}
  \begin{aligned}
    \forall v \in \V:
    \quad
    U &= \begin{cases}
          \, \set{v} \cup H(v) & \mbox{if } H(v) \neq \emptyset\\
          \, \set{v} \cup P(v) & \mbox{else}
        \end{cases}\\
    \Q(v) &= \qset(v, (U, \emptyset, \optithresh{U}))
  \end{aligned}
\end{equation}

Our results show that improvements to the naive case are possible when incorporating strategic considerations,
despite the fact that the quorum structure is heavily influenced by individual preferences.
More prominently---top tiers become of more manageable size (both for analysis and for consensus protocols leveraging the FBAS).

\begin{figure}[htpb]
  \begin{subfigure}[t]{\linewidth}
    \centering
    \includegraphics[width=\linewidth]{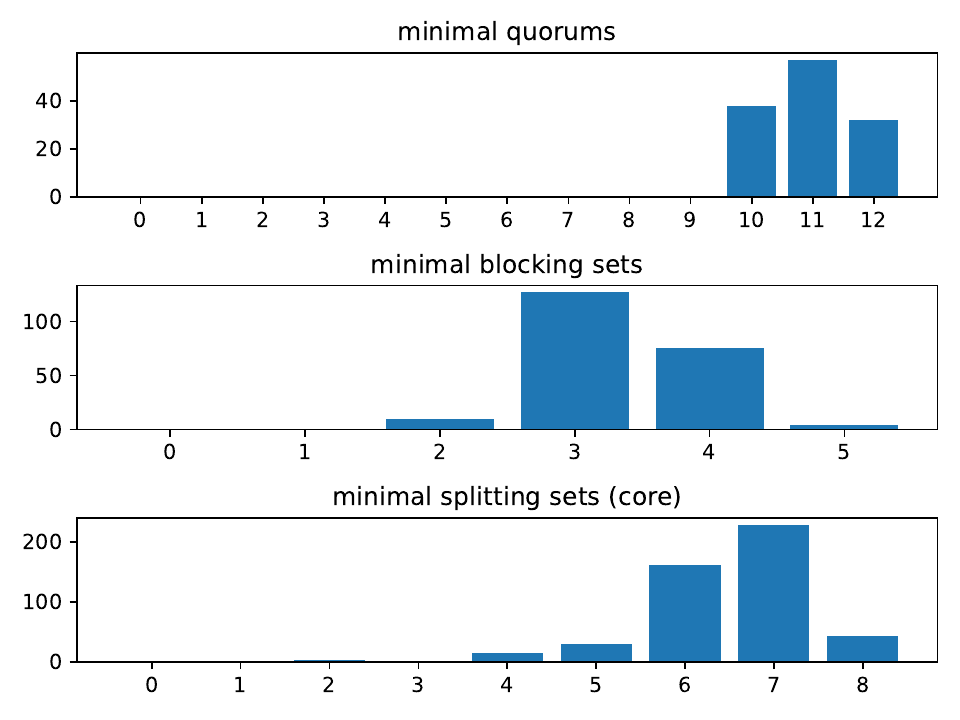}
    \caption{From $\oldASG$. Resulting FBAS has $\len{T} = 15$.}
    \label{fig:higher_tier_as98}
  \end{subfigure}
  \begin{subfigure}[t]{\linewidth}
    \centering
    \includegraphics[width=\linewidth]{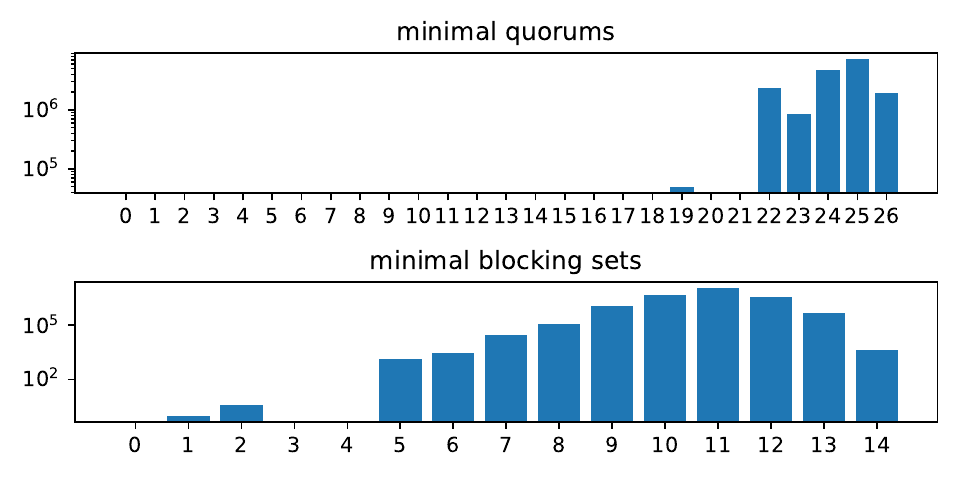}
    \caption{From $\newASG$. Resulting FBAS has $\len{T} = 36$.}
    \label{fig:higher_tier_as20}
  \end{subfigure}
  \caption{Histogram of the cardinalities of relevant sets in FBASs resulting from the application of \ref{qsc:higher_tier}
    using snapshots of the AS relationship graph ($\oldASG$, $\newASG$).}
  \label{fig:higher_tier}
\end{figure}

We simulated the application of \ref{qsc:higher_tier} using the AS graph snapshots $\oldASG$ and $\newASG$.
The two thus induced FBASs contained, respectively,
$2$ and $6$ nodes with one-node quorums sets which we filter our for the subsequent analysis.
We apply \texttt{fbas\_analyzer}, our software-based analysis framework (cf. \cref{sec:algorithms}), to the resulting FBASs.

\Cref{fig:higher_tier} presents the analysis findings.
It depicts histograms of the relevant sets,
i.e., how many minimal quorums, minimal blocking sets or minimal splitting sets of a given size exist for the given FBAS.
For the $\oldASG$ case,
we restricted our minimal splitting sets analysis to the core of the FBAS,
i.e.,
to its top tier and all nodes that are referenced by top tier nodes either directly or transitively\footnote{
  This corresponds to the union of all strongly connected components
  that contain a quorum.
}.
We find that doing so yields more informative results;
the full FBAS contains a large number of splitting sets with cardinality 1 that only split off very small groups of nodes from the rest.
Even when restricting the analysis to core nodes only,
we were not able to fully enumerate the minimal splitting sets for $\newASG$ in reasonable time,
due to the size and specific structure of the resulting FBAS.

Strikingly, our analysis reveals that the liveness of both FBASs is easily compromised.
Despite their relatively large top tiers (of $15$ and $36$ nodes, respectively),
groups of only 2 nodes, and in the $\newASG$ case even one group of only one node,
exist that are sufficient to completely block (or censor) the FBAS.
For comparison, symmetric top tiers of the same size would result in all minimal blocking sets having sizes of, respectively, $5$ and $12$.
This liveness-threatening discrepancy can be explained through \emph{cascading failures}:
If (for example) two nodes fail,
this can result in a third node with a \enquote{weak} quorum set becoming unsatisfiable,
so that three nodes have now de-facto failed,
which can result in a fourth node becoming unsatisfiable,
et cetera.
It can be concluded that the composition and size of smallest blocking sets for an FBAS is heavily influenced by the \enquote{weakest} quorum sets in the FBAS' top tier.
An additional example for cascading failures is given \cref{app:cascading_example}.

\subsection{Symmetry enforcement}%
\label{sub:symmetric_qsc}

The graph-based QSC policies discussed so far easily result in systems that are brittle
(in the sense of small minimal blocking sets)
and hard to analyze.
Both of these characteristics are vastly improved,
relative to top tier size,
in FBASs with symmetric top tiers.
However, symmetric top tiers emerge organically from a preexisting relationship graph $G$ only if the top tier nodes
form a complete subgraph of $G$,
which is not the case in the graphs investigated so far.
As a policy enhancement, nodes believing themselves to be top tier can \emph{mirror} the quorum sets of other apparently top tier nodes,
strategically including non-neighbors in their quorum sets for improving the global FBAS structure.
A behavior along this lines can, in fact, be observed in the live Stellar network (s.a. \cref{app:stellar_example}).

Yet, by making validator decisions independent of the local knowledge representation $G$,
new assumptions become necessary to be able to rule out attacks.
Mirroring makes it easier for malicious top tier nodes to introduce Sybil nodes into the top tier.
The approach is therefore only secure (w.r.t. both safety and liveness) if it can be assumed that nodes in $T$ make plausibility checks before expanding their quorum sets,
so that attempted (Sybil) attacks can be detected.
Given the lack of explicit incentives for running validator nodes in systems like Stellar,
such a burden on the operators of top tier nodes might be viewed as problematic \cite{kim2019stellar_secure}.
However, similar critique can also be voiced against systems (like Bitcoin) that base their security arguments on notions of economic rationality, as economic rationality can also be leveraged by attackers \cite{ford2019rationality}.

\section{Limits on openness and top tier fluidity}%
\label{sec:openness}

The FBAS paradigm reportedly enables the instantiation of consensus systems with
\emph{open membership}~\cite{mazieres2015stellar,lokhava2019stellar_payments}.
And clearly, arbitrary nodes can join an FBAS, causing new quorums to be formed that contain them.
Based on the preceding discussion, however,
we recognize that without creating a new, de-facto disjoint FBAS,
or the active reconfiguration of existing nodes,
new nodes cannot become part of \emph{minimal} quorums and hence minimal blocking sets.
Thereby,
their existence is irrelevant as far as the discussed liveness indicators are concerned,
and their importance for safety is limited.
In \cref{sec:analysis} we defined the notion of a \emph{top tier} to reflect the set of nodes in
an FBAS that is central to liveness, i.e.,
the set of nodes from which all minimal quorums and blocking sets are formed.
The top tier wields absolute power to censor and block the whole FBAS.

In the following, we investigate the question to what extent this top tier can be considered a group with open membership.
How can its power be diluted by promoting additional nodes to top tier status?
Can nodes be \enquote{fired} from the top tier?
We make the case that, in general,
a top tier $T$ can neither grow nor shrink without either the active involvement of existing top tier nodes \emph{or} a loss of safety guarantees.
We base all subsequent projections on the status quo of an FBAS that enjoys quorum intersection despite faulty nodes
(a \emph{safe} FBAS as per the discussion in \cref{sub:intersection}).

\subsection{Top-down top tier change}%
\label{sub:top_down_top_tier_change}

As a preliminary remark, recall that, as per \cref{def:top_tier},
we define the top tier $T$ of an FBAS $(\V, \Q)$ as the union of all its minimal quorums.
$T$ is therefore also a quorum and intersects every quorum in $(\V, \Q)$.

\begin{theorem}[top tier can safely change itself]
  \label{theorem:top_down_change}
  Let $T \subset \V$ be the top tier of an FBAS $(\V, \Q)$ that enjoys quorum availability and quorum intersection.
  Then it is possible,
  without compromising neither quorum availability nor quorum intersection,
  to instantiate a new top tier $T^\prime \subseteq \V, T^\prime \neq \emptyset$
  by changing only the quorum sets of new and old top tier nodes $v \in T \cup T^\prime$.
\end{theorem}
\begin{proof}
  Let $T^\prime \subseteq \V, T^\prime \neq \emptyset$ be the target top tier.
  Let $\Q^\prime$ be a modification of $\Q$ so that
  $\forall v \in T \cup T^\prime: \Q^\prime(v) = \set{T^\prime}$\footnote{
    Without loss of generality. Clearly, more robust top tier constructions are possible.
  }
  and
  $\forall v \notin T \cup T^\prime: \Q^\prime(v) = \Q(v)$.
  As $T^\prime$ is a quorum w.r.t. $\Q^\prime$, $(T^\prime, \Q^\prime)$ enjoys quorum availability.
  Therefore, $(\V, \Q^\prime)$ enjoys quorum availability.
  $(\V \setminus T', \Q^\prime)$ does not enjoy quorum availability,
  because no node in $T$ is satisfied without $T^\prime$ and
  no node in $\V \setminus T$ can form a quorum without a node from $T$
  (otherwise $T$ would not have been the top tier w.r.t. $\Q$, cf. \cref{def:top_tier}).
  There are therefore no quorums w.r.t. $\Q^\prime$ that are disjoint of $T^\prime$.
  $(\V, \Q^\prime)$ therefore enjoys quorum intersection iff $(T^\prime, \Q^\prime)$ enjoys quorum intersection,
  which it (trivially) does.
  \qed
\end{proof}

The situation is less clear if some nodes $T \setminus T^\prime$ do not wish to leave $T$.
Note, however,
that single nodes can always endanger safety via trivial configurations such as $\Q(v) = \set{\set{v}}$.
If performed by one or more nodes in $T$,
such an act of sabotage can
have an impact on the safety of large portions of the FBAS.

\subsection{Bottom-up top tier change}%
\label{sub:bottom_up_top_tier_change}

In the following,
we assume a \enquote{self-centered} top tier in the sense that all top tier nodes include only other top tier nodes in quorum sets.
Symmetric top tiers (\cref{def:symmetric_top_tier}) have this property,
as do top tiers observed in the wild in the Stellar network (cf. \cref{app:stellar_example}).

\begin{theorem}[no safe top tier change with uncooperative top tier]
  \label{theorem:bottom_up_growth}
  Let $(\V, \Q)$ be an FBAS that enjoys quorum intersection and has
  a \enquote{self-centered} top tier $T \subset \V$ such that all top tier quorum slices are comprised of only top tier nodes
  ($\forall v \in \V : \bigcup{\Q(v)} \subseteq T$).
  Then it is \textbf{not} possible,
  without compromising quorum intersection,
  to instantiate a new top tier $T^\prime \subseteq \V, T^\prime \neq T$
  by changing only the quorum sets of non-top tier nodes $v \in \V \setminus T$.
\end{theorem}
\begin{proof}
  Let $T^\prime \subseteq \V, T^\prime \neq T$ be the top tier of a new FBAS $(\V, \Q^\prime)$ that enjoys quorum intersection.
  Let $\Min{\Quorums}$ and $\Min{\Quorums}^\prime$ be the sets of all minimal quorums of $(\V, \Q)$ and $(\V, \Q^\prime)$, respectively.
  As per \cref{def:top_tier}, $T^\prime \neq T$ implies that $\Min{\Quorums} \neq \Min{\Quorums}^\prime$.

  Assume there exists a $\Min{U} \in \Min{\Quorums} \setminus \Min{\Quorums}^\prime$.
  Then
  $\Min{U}$ is a quorum w.r.t. $\Q$
  and either (a) not a quorum w.r.t. $\Q^\prime$
  or (b) not minimal w.r.t. $\Q^\prime$.
  However, we require that the quorum sets of top tier nodes don't change: $\forall v \in T: \Q^\prime(v) = \Q(v)$.
  Therefore $\Min{U}$ is a quorum also w.r.t. $\Q^\prime$, contradicting (a).
  Hence, (b) must hold and there must be a
  $\Min{U}^\prime \in \Min{\Quorums}^\prime$ such that $\Min{U}^\prime \subset \Min{U}$ (cf. \cref{def:minimal_set}).
  As $\Min{U}^\prime \subseteq \Min{U} \subseteq T$, $\Min{U}^\prime$ being a quorum w.r.t. $\Q^\prime$ implies it also being a quorum w.r.t. $\Q$.
  But then $\Min{U}$ is not minimal w.r.t. $\Q$, implying $\Min{U} \notin \Min{\Quorums}$ and thus again leading to a contradiction.
  This proves that $\Min{\Quorums} \subseteq \Min{\Quorums}^\prime$.

  Assume now there exists a $\Min{U}^\prime \in \Min{\Quorums}^\prime \setminus \Min{\Quorums}$ and let $\Min{U} \in \Min{\Quorums}$.
  As $(\V, \Q^\prime)$ enjoys quorum intersection,
  $\Min{U}^\prime \cap \Min{U} \neq \emptyset$ and $\Min{U}^\prime$ contains members of the \enquote{old} top tier $T$.
  $\Min{U}^\prime$ is a quorum w.r.t. $\Q^\prime$,
  but $\Min{U}^\prime \cap T$ cannot be a quorum w.r.t. $\Q^\prime$ as otherwise $\Min{U}^\prime$ would not be a minimal quorum.
  There must therefore exist a node $v \in \Min{U}^\prime \cap T$ with a quorum slice $q \in \Q^\prime(v)$ such that
  $(\Min{U}^\prime \cap T) \subset q \subseteq \Min{U}^\prime$ (cf. \cref{def:quorum}),
  i.e., $q \setminus T \neq \emptyset$.
  As $v \in T$, we require that $\Q^\prime(v) = \Q(v)$ and $\bigcup{\Q(v)} \subseteq T$,
  which leads to a contradiction since $q \in \Q(v)$ and $q \setminus T \neq \emptyset$.
  It must therefore hold that $\Min{\Quorums} \setminus \Min{\Quorums}^\prime = \emptyset$, $\Min{\Quorums}=\Min{\Quorums}^\prime$ and $T = T^\prime$.
  \qed
\end{proof}

\subsection{Consequences}%
\label{sub:discussion}

Who determines which FBAS nodes get to form the top tier?
Our results imply that,
if maintaining safety is seen as an untouchable requirement,
the top tier $T_i$ of an FBAS $(\V_i, \Q_i)$ at \enquote{iteration} $i$ is legitimated by decisions of, exclusively, members of $T_{i-1} \cup T_i$
(if none of them cooperates, we lose safety, if all of them cooperate, we don't).
Because of the top tier's importance to the liveness, safety and performance achievable within a given FBAS,
open membership in $\V_i$ is of little benefit without open membership in $T_i$.

How closed is the membership in $T_i$?
It might be sufficient that only some nodes in $T_{i-1}$ support a transition to $T_i$.
If reactive QSC policies are used (e.g., for enforcing top tier symmetry as discussed in \cref{sub:symmetric_qsc}),
one cooperative top tier node $v \in T_{i-1}$ might already be enough for growing the top tier
in a way that is robust and doesn't only dilute the relative influence of $v$.
How partially supported top tier changes would play out must be investigated based on more specific scenarios.
We expect the safe \enquote{firing} of top tier nodes to be especially challenging.

Which begs the question---can the safety requirement
be weakened?
For example,
given sufficiently good (out-of-band) coordination between members of $\V_{i-1} \setminus T_{i-1}$,
a $(\V_i, \Q_i)$ might be instantiated in which at least $(\V_i \setminus T_{i-1}, \Q_i)$ enjoys quorum intersection.
It is conceivable that novel protocols can be developed,
possibly also leveraging the FBAS structure,
that reduce the notorious difficulty of coordinating such bottom-up actions.

\section{Conclusion}
\label{sec:conclusion}

We demonstrate in this paper that,
despite the complexity of the FBAS model,
the properties of concrete FBAS instances can be described in a way that is
both precise and intuitive, and allows comparisons with more classical Byzantine agreement systems.
We propose the notions of minimal blocking sets, minimal splitting sets and top tiers to describe
which groups of nodes can compromise liveness and safety.
In essence, minimal blocking sets and minimal splitting sets describe minimal viable threat scenarios,
thereby enabling a comprehensive risk assessment in FBAS-based systems like the Stellar network.
While some analyses imply computational problems of exponential complexity,
we developed and implemented algorithms that enable the exact analysis of a wide range of interesting FBASs.

Our implemented analysis framework also enables us to investigate how individual configurations result in global properties.
We find that overly strategic configuration policies result in FBASs that are indistinguishable from permissioned systems.
Individualistic approaches, on the other hand, cannot guarantee safe results while quickly resulting in systems that are infeasible to analyze.
Adding some strategic decision-making at organically emerging top tier nodes
offers a potential middle way towards robust FBASs instantiated from the sum of individual preferences.

Independently of the way in which a given FBAS came to be, however,
the composition of a once established top tier cannot be influenced without the cooperation of existing top tier nodes,
without at the same time threatening safety.
This seems to place the FBAS paradigm closer to the \enquote{permissioned consensus} camp than hoped.
More investigation is needed to determine the exact impact of bottom-up top tier changes
(as in number of nodes affected by a loss of safety or liveness, for example)
and to formulate possible coordination strategies to keep such impacts low.

\bibliographystyle{spmpsci}
\bibliography{paper}

\newpage

\appendix

\section{Additional corollaries, theorems and proofs}%
\label{app:proofs}

\subsection{Minimal quorums}%
\label{appsub:minimal_quorums}

\begin{corollary}[minimal quorum intersection $\iff$ quorum intersection]
  \label{cor:min_quorum_intersection}
  Let $\Quorums \subseteq 2^{\V}$ be the set of all quorums of the FBAS $(\V, \Q)$,
  $\Min{\Quorums} \subseteq \Quorums$ be the set of all minimal quorums.
  All pairs of $U_1, U_2 \in \Quorums$ intersect iff all pairs of $\Min{U_1}, \Min{U_2} \in \Min{\Quorums}$ intersect.
\end{corollary}
\begin{proof}
  Since $\Min{\Quorums} \subseteq \Quorums$,
  $\forall U_1, U_2 \in \Quorums : U_1 \cap U_2 \neq \emptyset$
  trivially implies that
  $\forall \Min{U}_1, \Min{U}_2 \in \Min{\Quorums} : \Min{U}_1 \cap \Min{U}_2 \neq \emptyset$.
  The other direction follows because
  $\forall U_1, U_2 \in \Quorums \; \exists \Min{U}_1, \Min{U}_2 \in \Min{\Quorums} :
  \Min{U}_1 \subseteq U_1 \land \Min{U}_2 \subseteq U_2$
  ($\Min{\Quorums}$ being the set of all minimal sets w.r.t. $\Quorums$; s.a. \cref{def:minimal_set}).
  If all pairs in $\Min{\Quorums}$ intersect,
  so must therefore all pairs in $\Quorums$.
  \qed
\end{proof}

This was previously also shown in \cite{lachowski2019complexity}.

\subsection{Blocking sets}%
\label{appsub:blocking_sets}

\begin{corollary}[blocking for all $\implies$ blocking for all minimal]
  \label{cor:block_all_minimal}
  Let $\Quorums \subseteq 2^{\V}$ be the set of all quorums of the FBAS $(\V, \Q)$,
  and $\Min{\Quorums} \subseteq \Quorums$ be the set of all minimal quorums.
  If $B$ is a blocking set for $\Quorums$,
  then it is also a blocking set for $\Min{\Quorums}$.
\end{corollary}
\begin{proof}
  $B$ is a blocking set for
  $\Quorums \iff \forall U \in \Quorums : B \cap U \neq \emptyset$ (\cref{def:blocking_set}).
  $\Min{\Quorums} \subseteq \Quorums \implies \forall \Min{U} \in \Min{\Quorums} : B \cap
  \Min{U} \neq \emptyset$, so that $B$ is also a blocking set for $\Min{\Quorums}$.
  \qed
\end{proof}

\begin{corollary}[blocking for all minimal $\implies$ blocking for all]
  \label{cor:block_minimal_all}
  Let $\Quorums \subseteq 2^{\V}$ be the set of all quorums of the FBAS $(\V, \Q)$,
  and $\Min{\Quorums} \subseteq \Quorums$ be the set of all minimal quorums.
  If $B$ is blocking set for $\Min{\Quorums}$,
  then it is also a blocking set for $\Quorums$.
\end{corollary}
\begin{proof}
  $B$ is a blocking set for
  $\Min{\Quorums} \implies \forall U \in \Min{\Quorums} : B \cap U \neq \emptyset$ (\cref{def:blocking_set}).
  $\Min{\Quorums} \subseteq \Quorums$ and all $U \in \Min{\Quorums}$ are minimal w.r.t. $\Quorums$
  $\implies \forall U \in \Quorums \; \exists \Min{U} \in \Min{\Quorums} : \Min{U} \subseteq U$ (cf. \cref{def:minimal_set})
  $\implies U \cap B \neq \emptyset \implies$ $B$ is blocking for all $U \in \Quorums$.
  \qed
\end{proof}

\begin{corollary}[minimal blocking sets result from minimal quorums]
  \label{cor:min_block_from_min_quorum}
  Let $\Quorums \subseteq 2^{\V}$ be the set of all quorums of the FBAS $(\V, \Q)$,
  $\Min{\Quorums} \subseteq \Quorums$ be the set of all minimal quorums,
  and $\Min{\BlockSetSet} \subseteq 2^{\V}$ be the set of all minimal blocking sets.
  Then each minimal blocking set $\Min{B} \in \Min{\BlockSetSet}$ of the FBAS
  is minimally blocking w.r.t. $\Min{\Quorums}$, i.e., $\Min{B}$ intersects every minimal quorum $\Min{U} \in \Min{\Quorums}$
  and no $B^\prime \subset \Min{B}$ intersects every minimal quorum $\Min{U} \in \Min{\Quorums}$.
\end{corollary}
\begin{proof}
  Let $\BlockSetSet \subseteq 2^{\V}$ be the set of all blocking sets w.r.t. $\Min{\Quorums}$.
  Based on \cref{cor:block_all_minimal} and \cref{cor:block_minimal_all},
  $\BlockSetSet$ is exactly the set of all blocking sets for $\Quorums$.
  Hence the set of all minimal sets w.r.t. $\BlockSetSet$
  is exactly the set of all minimal blocking sets w.r.t. $\Quorums$
  and therefore the set of all minimal blocking sets for $(\V, \Q)$,
  or $\Min{\BlockSetSet} \subseteq \BlockSetSet$.
  Likewise, as $\BlockSetSet$ is the set of all blocking sets w.r.t. $\Min{\Quorums}$,
  $\Min{\BlockSetSet}$ is the set of all minimal blocking sets w.r.t. $\Min{\Quorums}$.
  \qed
\end{proof}

\subsection{Splitting sets}%
\label{appsub:splitting_sets}

\begin{definition}[quorum expanders]
  \label{def:quorum_expanders}
  For an FBAS $(\V, \Q)$, a \emph{quorum expander} is any
  node $v \in \V$ that is part of a quorum slice $q \in \Q(v^\prime)$ of another node $v^\prime \in \V$
  that is a not a quorum slice for $v$,
  i.e., any node $v \in \V$ for which
  $\exists v^\prime \in \V, q^\prime \in \Q(v^\prime):
  v \in q^\prime \land (\forall q \in \Q(v): q \not\subseteq q^\prime)$.
\end{definition}

\begin{theorem}[minimal splitting sets formed exclusively of quorum expanders and top tier nodes]
  \label{theorem:minimal_splitting_set_nodes_are_quorum_expanders_or_top_tier}
  Let $\Min{\SplitSetSet} \subseteq 2^{\V}$ be the set of all minimal splitting sets of the FBAS $(\V, \Q)$,
  $X \subseteq \V$ the set of all quorum expanders of the FBAS (\cref{def:quorum_expanders})
  and
  $T \subseteq \V$ the top tier of the FBAS
  (the union of all minimal quorums, \cref{def:top_tier}).
  Then it holds that $\bigcup \Min{\SplitSetSet} \subseteq T \cup X$.
\end{theorem}
\begin{proof}
  Let $\Min{S} \in \Min{\SplitSetSet}$ and $s \in \Min{S}$ be an arbitrary node in that splitting set.
  We show that $s \in T$ or $s \in X$ must hold.

  $\Min{S}$ is a minimal splitting set, therefore $\Min{S} \setminus \set{s}$ is not a splitting set for any $s$.
  Consequently,
  $(\V, \Q)^{\Min{S}\setminus{\set{s}}}$ enjoys quorum intersection while $(\V, \Q)^{\Min{S}}$ doesn't.
  Let $\Min{U}_1, \Min{U}_2 \subset \V, \Min{U}_1 \cap \Min{U}_2 = \emptyset$
  be two non-intersecting minimal quorums in $(\V, \Q)^{\Min{S}}$
  such that $\Min{U}_1$ does not contain a quorum in $(\V, \Q)^{\Min{S}\setminus{\set{s}}}$.
  (If both $\Min{U}_1$ and $\Min{U}_2$ contained quorums in $(\V, \Q)^{\Min{S}\setminus{\set{s}}}$, the FBAS would lack quorum intersection.)

  If $\Min{U}_1 \cup \set{s}$ contains a quorum in $(\V, \Q)^{\Min{S}\setminus{\set{s}}}$,
  then $\Min{U}_1 \cup \set{s}$ contains a minimal quorum $\Min{U}_1^\prime \subseteq \Min{U}_1 \cup \set{s}$ that contains $s$.
  Consequently, $s$ is part of the top tier $T^\prime$ of $(\V, \Q)^{\Min{S}\setminus{\set{s}}}$,
  i.e., $s \in T^\prime$.
  As the only effect of the delete operation (\cref{def:delete}) on $\Q$ is to remove nodes from quorum slices
  and both $(\V, \Q)$ and $(\V, \Q)^{\Min{S}\setminus{\set{s}}}$ enjoy quorum intersection,
  it holds that $T^\prime \subseteq T$
  (the proof is analogous to the proof of \cref{theorem:bottom_up_growth}).
  Consequently, $s \in T$.

  If $\Min{U}_1 \cup \set{s}$ does not contain a quorum in $(\V, \Q)^{\Min{S}\setminus{\set{s}}}$,
  then,
  because $\Min{U}_1$ is a quorum in $(\V, \Q)^{\Min{S}}$,
  the forming of a quorum fails because of $s$.
  For $(\V^\prime, \Q^\prime) := (\V, \Q)^{\Min{S}\setminus{\set{s}}}$,
  it must hold that $\exists v \in \Min{U}_1, \exists q \in \Q^\prime(v) : q \subseteq \Min{U}_1 \cup \set{s}$
  while $\forall q^\prime \in \Q^\prime(s) : q \not\subseteq \Min{U}_1 \cup \set{s}$.
  The node $s$ is therefore one of the quorum expanders $X^\prime$ of $(\V, \Q)^{\Min{S}\setminus{\set{s}}}$,
  i.e., $s \in X^\prime$.
  It trivially holds that $X^\prime \subseteq X$ and,
  therefore,
  $s \in X$.

  \qed
\end{proof}

\subsection{Top tier}%
\label{appsub:top_tier}

\begin{corollary}[minimal blocking sets formed exclusively of top tier nodes]
  \label{cor:min_block_from_top_tier}
  Let $T$ be the top tier of an FBAS $(\V, \Q)$,
  and $\Min{\BlockSetSet} \subseteq 2^{\V}$ be the set of all minimal blocking sets of $(\V, \Q)$.
  Then $\forall \Min{B} \in \Min{\BlockSetSet} : \Min{B} \subseteq T$.
\end{corollary}
\begin{proof}
  From \cref{cor:min_block_from_min_quorum}
  it follows that
  all $\Min{B} \in \Min{\BlockSetSet}$
  are formed of nodes contained in at least one minimal quorum
  $\Min{U} \in \Min{\Quorums}$.
  As $T = \bigcup{\Min{\Quorums}}$ (\cref{def:top_tier}),
  $\forall \Min{B} \in \Min{\BlockSetSet} : \Min{B} \subseteq T$.
  \qed
\end{proof}

\begin{theorem}[each top tier node in at least one minimal blocking set]
  \label{theorem:top_tier_from_min_block}
  Let $T$ be the top tier of an FBAS $(\V, \Q)$,
  and $\Min{\BlockSetSet} \subseteq 2^{\V}$ be the set of all minimal blocking sets of $(\V, \Q)$.
  Then for each top tier node $v \in T$ there is at least one minimal blocking set
  $\Min{B} \in \Min{\BlockSetSet}$ such that $v \in \Min{B}$.
\end{theorem}
\begin{proof}
  Let $v \in T$ be an arbitrary top tier node and
  $\Min{U} \in \Min{\Quorums}$ an arbitrary minimal quorum such that $v \in \Min{U}$
  (recall that $T = \bigcup{\Min{\Quorums}}$; \cref{def:top_tier}).
  $T \setminus \Min{U}$ intersects every $\Min{U}^\prime \in \Min{\Quorums} \setminus \set{\Min{U}}$,
  as otherwise there would be a $\Min{U}^\prime \in \Min{\Quorums}$ such that $\Min{U}^\prime \subset \Min{U}$
  (i.e., $\Min{U}$ would not be a minimal quorum).
  Therefore,
  $T \setminus \Min{U}$ is a blocking set w.r.t. $\Min{\Quorums} \setminus \set{\Min{U}}$ and
  $B^\prime = \set{v} \cup T \setminus \Min{U}$ is a blocking set w.r.t. $\Min{\Quorums}$.
  $B^\prime \setminus \set{v}$ is not a blocking set w.r.t. $\Min{\Quorums}$ because it doesn't intersect $\Min{U}$.
  Hence, all $\Min{B} \in \Min{\BlockSetSet}$ such that $\Min{B} \subseteq B^\prime$
  (and there must be at least one---$B^\prime$---because $B^\prime$ is a blocking set w.r.t. $\Min{\Quorums}$)
  must contain $v$.
  Hence the FBAS has at least one
  $\Min{B} \in \Min{\BlockSetSet}$ that contains $v$.
  \qed
\end{proof}

\begin{theorem}[Bocking sets in non-nested symmetric top tier]
  \label{theorem:symmetric_top_tier_blocking_cardinalities}
  For an FBAS $(\V, \Q)$ with a symmetric top tier $T \subseteq \V$, $\tts := |T|$
  such that $\forall v \in T: \Q(v) = \qset(v, (T, \emptyset, t))$ it holds that:
  All minimal blocking sets $\Min{B} \in \Min{\BlockSetSet}$ have cardinality $\max(\tts - t + 1, 0)$.
\end{theorem}
\begin{proof}
  We observe that for any $v \in T$,
  $\Q(v) = \set{q \subseteq \V : v \in q \land \len{q \cap T} \geq t}$ (\cref{def:quorum_set,def:qset_function}).
  A $U \subset T$ is therefore a quorum in $(\V, \Q)$ iff $\len{U} \geq t$ (\cref{def:quorum}).
  As all $U \subset T$ with $\len{U} \geq t$ are quorums in $(\V, \Q)$,
  the minimal quorums in $(\V, \Q)$ are exactly
  $\Min{\Quorums} = \set{\Min{U} \subseteq T, \len{\Min{U}} = t}$.
  Then:

  For all $B \subseteq T$ with $\len{B} = \tts - t + 1$ it holds that $\forall U^\prime \subseteq T \setminus B : \len{U^\prime} = t - 1 < t$.
  Hence, no $U^\prime \subseteq T \setminus B$ is a quorum, there are no quorums that are disjoint with $B$
  and $B$ is a blocking set (\cref{def:blocking_set}).
  $B$ is furthermore a minimal blocking set,
  as for any $B^\prime \subset B$ it holds that $U = T \setminus B^\prime$ is a quorum (as $\len{U} \geq t$),
  and so $B^\prime$ is not a blocking set.

\end{proof}

\begin{theorem}[Splitting sets in non-nested symmetric top tier]
  \label{theorem:symmetric_top_tier_splitting_cardinalities}
  For an FBAS $(\V, \Q)$ that consists entirely of a symmetric top tier $T = \V$, $\tts := |T|$
  such that
  $\forall v \in \V: \Q(v) = \qset(v, (\V, \emptyset, t))$ it holds that
  all minimal splitting sets $\Min{S} \in \Min{\SplitSetSet}$ have cardinality $\max(2t - \tts, 0)$.
\end{theorem}
\begin{proof}
  Like in \cref{theorem:symmetric_top_tier_blocking_cardinalities},
  we observe that
  the minimal quorums in $(\V, \Q)$ are exactly
  $\Min{\Quorums} = \set{\Min{U} \subseteq T, \len{\Min{U}} = t}$.
  Then:

  Let $\Min{S} \in \Min{\SplitSetSet}$ be an arbitrary minimal splitting set for $(\V, \Q)$.
  If $2t - m \leq 0$, there exist two minimal quorums $\Min{U}_1, \Min{U}_2 \in \Min{\Quorums}$ (with cardinality $t$)
  that do not intersect.
  There is then only one $\Min{S} = \emptyset$ and the cardinality of all minimal splitting sets is trivially $0$.
  In the following, we assume that $2t - m > 0$ and $(\V, \Q)$ therefore enjoys quorum intersection.
  Since $(\V, \Q)$ consists entirely of a symmetric top tier, no $v \in \V$ is a quorum expander.
  Splitting sets must therefore contain an intersection of at least one pair of minimal quorums
  (for illustration, cf. the proof of \cref{theorem:minimal_splitting_set_nodes_are_quorum_expanders_or_top_tier}).
  There are therefore at least two minimal quorums
  $\Min{U}_1, \Min{U}_2 \in \Min{\Quorums}$ such that $\Min{S} = \Min{U}_1 \cap \Min{U}_2$.
  Let $U = \Min{U}_1 \cup \Min{U}_2$.
  $N^\prime = T \setminus U$ must be empty,
  otherwise we could,
  with an arbitrary $N^{\prime\prime} \subseteq \Min{S},\len{N^{\prime\prime}} = \len{N^\prime}$
  find a minimal quorum $\Min{U}_3 = (\Min{U}_2 \setminus N^{\prime\prime}) \cup N^\prime$
  such that $\Min{U}_1 \cap \Min{U}_3 \subset \Min{S}$
  (i.e., $\Min{S}$ is not minimal).
  It therefore holds that $U = T$ and,
  since,
  $\len{\Min{U}_1} = \len{\Min{U}_2} = t$,
  $\len{\Min{S}} = 2t - m$.
  \qed
\end{proof}

\section{Example analysis: toy network with cascading failures}%
\label{app:cascading_example}

Consider the FBAS $(\V, \Q)$ with $\V = \set{0, 1, 2, 3, 4, 5, 6}$ and $\Q$ such that:
\begin{align*}
  \Q(0) &= \qset(0, (\set{0, 1, 2}, \emptyset, 3))\\
  \Q(1) &= \qset(1, (\set{0, 1, 2, 3}, \emptyset, 3))\\
  \Q(2) &= \qset(2, (\set{0, 1, 2, 3, 4, 5, 6}, \emptyset, 5))\\
  \Q(3) &= \qset(3, (\set{0, 1, 2, 3, 4, 5, 6}, \emptyset, 5))\\
  \Q(4) &= \qset(4, (\set{0, 1, 2, 3, 4, 5, 6}, \emptyset, 5))\\
  \Q(5) &= \qset(5, (\set{0, 1, 2, 3, 4, 5, 6}, \emptyset, 5))\\
  \Q(6) &= \qset(6, (\set{0, 1, 2, 3, 4, 5, 6}, \emptyset, 5))\\
\end{align*}

This $\Q$ can be the result of a scenario in which all $v \in \V$ apply the QSC policy
\ref{qsc:all_neighbors} (\cref{sub:individual_qsc}) based on following graph $G$
(unidirectional edges highlighted as dashed lines):

\begin{center}
  \begin{tikzpicture}
    \def \radius {1.5cm}

    \foreach \s in {0,...,6}
    {
      \node[draw, circle] (v\s) at ({135 - 360/7 * \s}:\radius) {$\s$};
    }
    \path [<->] (v0) edge [draw] (v1);
    \path [<->] (v0) edge [draw] (v2);
    \path [->, dashed]  (v3) edge [draw] (v0);
    \path [->, dashed]  (v4) edge [draw] (v0);
    \path [->, dashed]  (v5) edge [draw] (v0);
    \path [->, dashed]  (v6) edge [draw] (v0);

    \path [<->] (v1) edge [draw] (v2);
    \path [<->] (v1) edge [draw] (v3);
    \path [->, dashed]  (v4) edge [draw] (v1);
    \path [->, dashed]  (v5) edge [draw] (v1);
    \path [->, dashed]  (v6) edge [draw] (v1);

    \path [<->] (v2) edge [draw] (v3);
    \path [<->] (v2) edge [draw] (v4);
    \path [<->] (v2) edge [draw] (v5);
    \path [<->] (v2) edge [draw] (v6);

    \path [<->] (v3) edge [draw] (v4);
    \path [<->] (v3) edge [draw] (v5);
    \path [<->] (v3) edge [draw] (v6);

    \path [<->] (v4) edge [draw] (v5);
    \path [<->] (v4) edge [draw] (v6);

    \path [<->] (v5) edge [draw] (v6);
  \end{tikzpicture}
\end{center}

We find the minimal blocking sets $\Min{\BlockSetSet} \subset 2^{\V}$ of $(\V, \Q)$ using our analysis tool (cf. \cref{sec:algorithms}):
\begin{align*}
  \Min{\BlockSetSet} = \set{&\set{2},\set{1,3},\set{1,4},\set{1,5},\set{1,6},\set{0,3},\set{3,4,5},\set{3,4,6},\\
                            &\set{3,5,6},\set{0,4,5},\set{0,4,6},\set{0,5,6},\set{4,5,6}}
\end{align*}

Despite the fact that most nodes in $\V$ have very \enquote{robust} quorum sets---%
being able to tolerate up to $f = 2$ failures,
which corresponds to a minimal blocking set of cardinality $3$---%
the smallest blocking set of $(\V, \Q)$, $\set{2}$, actually has cardinality $1$.
Consider a failure of node $2$.
Node $0$'s quorum set ($\Q(0)$) is not satisfiable anymore, so that $0$ de-facto fails as well.
With both $0$ and $2$ failed, node $1$, being able to tolerate only $f=1$ failures,
becomes unsatisfiable as well.
With three nodes having de-facto failed, none of the remaining nodes' quorum sets can be satisfied anymore,
so that $(\V, \Q)$ loses quorum availability.
Enabled through the \enquote{weak} quorum sets of nodes $0$ and $1$,
the failure of $2$ triggers what we would call a \emph{cascading failure}.
The liveness \enquote{buffer} of $(\V, \Q)$,
as represented by its smallest blocking sets,
is determined by the most easily dissatisfied nodes in its top tier.

We see a similar, although weaker effect with regards to minimal splitting sets.
In the present example, there are fewer minimal splitting sets
$\Min{\SplitSetSet} \subset 2^{\V}$
than in an \enquote{ideal} FBAS of the same size
(cf. \ref{qsc:ideal_open} in \cref{sub:qsc_policy})
but all but one of them have the \enquote{ideal} cardinality $3$ or a larger cardinality:
\begin{align*}
  \Min{\SplitSetSet} = \set{&\set{1,2},\set{0,1,3},\set{0,1,4},\set{0,2,3},\set{0,2,4},\set{0,3,4},\\
                            &\set{1,3,4,5},\set{2,3,4,5}}
\end{align*}

Note that unlike blocking sets that can compromise liveness for all nodes in an FBAS,
splitting sets are usually more relevant to some nodes than they are to others.
For example, the smallest splitting set of $(\V,\Q)$,
$\set{1,2}$,
can potentially cause node $0$ to diverge from the remainder of the network---this is likely a bigger problem for node $0$ than for nodes $\set{3,4,5}$
which would remain \enquote{in sync}.

\section{Example analysis: Stellar network}%
\label{app:stellar_example}

As an example for the results obtainable using the proposed methodology and tooling,
we will now present a short study into the Stellar FBAS~\cite{lokhava2019stellar_payments}\footnote{
  We maintain an interactive version of this study at: \url{https://trudi.weizenbaum-institut.de/stellar_analysis/}
}.
Our analysis methodology has furthermore been integrated into \emph{Stellarbeat}\footnote{
  \url{https://stellarbeat.io/}
},
a popular monitoring website for the Stellar network.

For the presented study, we obtain daily snapshots of the Stellar FBAS from Stellarbeat\footnote{
  Data from Stellarbeat was also used in previous academic studies such as \cite{kim2019stellar_secure}.
},
for the interval July 2019 -- January 2022.
From the same source, we also obtain data for allocating nodes,
here individual network hosts running the Stellar software,
to the \emph{organizations} they belong to.
We use this data to merge nodes belonging to the same organization,
so that nodes in the subsequent discussion represent distinct organizations as opposed to individual physical machines\footnote{
  Nodes can also be merged based on other criteria, such as their country or ISP, revealing different threat scenarios.
  For example, for a snapshot of the Stellar FBAS from November 2020,
  we determine that a certain large cloud hosting provider forms a blocking set---i.e.,
  has the power to unilaterally compromise liveness.
}.
For maintaining the correctness of our results, we merge nodes in this way \emph{after} completing the analyses.
\emph{Prior} to analysis,
we filter out all nodes that are marked as inactive or induce one-node quorums
(i.e., nodes $v$ with a configuration such as $\Q(v) = \set{v}$;
we assume that this represents an accidental misconfiguration).
We furthermore restrict our minimal splitting sets analyses to a core subset of nodes for each FBAS snapshot,
namely to the top tier and all nodes transitively referenced by top tier nodes' quorum sets.
Doing so gives us more informative aggregate results as
forming a splitting set that affects only a few edge nodes is both significantly easier and less impactful than forming a splitting set that can cause top tier nodes to diverge.
All analyses were performed using the algorithms and implementation introduced in \cref{sec:algorithms}.
The results of our study are presented in \cref{fig:stellar}.

\begin{figure}[htpb]
  \centering
  \includegraphics[width=\linewidth]{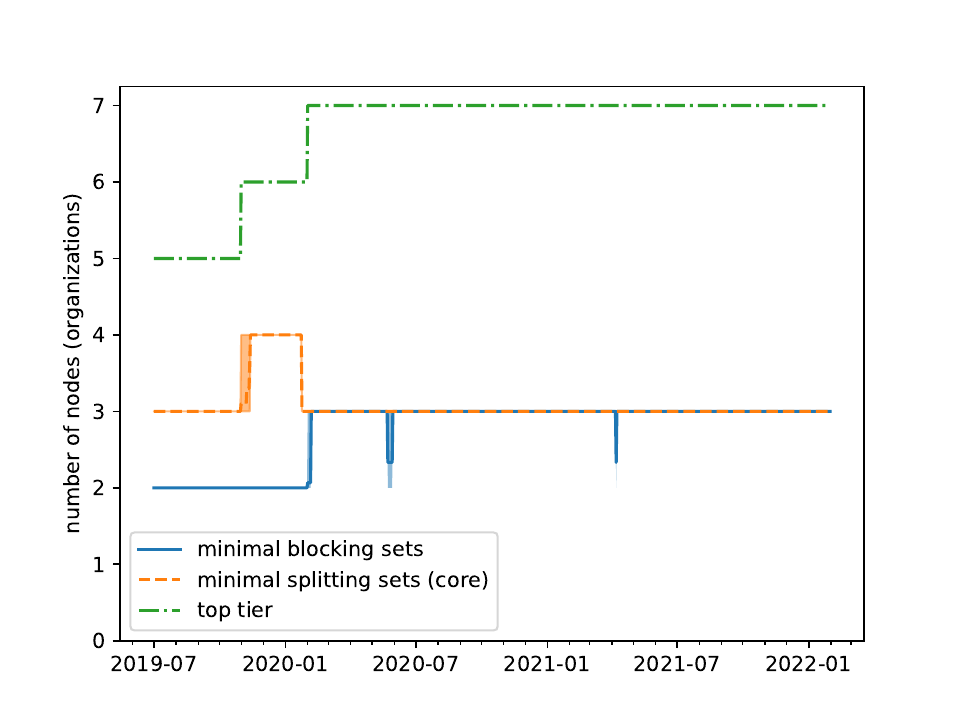}
  \caption{
    Analysis results for daily snapshots of the Stellar network.
    For each presented FBAS snapshot,
    the plot charts the size of its top tier as well as the mean cardinalities of minimal blocking and minimal splitting sets,
    with area boundaries marking the cardinalities of the smallest and largest respective set.
  }
  \label{fig:stellar}
\end{figure}

The top tier of the Stellar network is growing monotonically through time in the studied interval,
reaching $7$ organizations in February 2020.
The top tiers of most analyzed snapshots are symmetric and resemble (on the organizations level) a classical (non-nested) threshold-based quorum system.
In \cref{fig:stellar}, symmetric top tiers of such a type manifest themselves as data points in which the cardinalities of all minimal blocking sets are identical,
as are the cardinalities of all minimal splitting sets.
During February 2020, the top tier grew by one organization, disturbing the symmetry for a few days.
However, eventually all top tier nodes included the new organization into their quorum sets.
This adaptation suggests that top tier nodes might be reacting to each others' decisions and actively strive towards a symmetric configuration,
as proposed in \cref{sub:symmetric_qsc}.
Furthermore,
the thresholds of top tier quorum sets appear to be chosen based on a \SI{67}{\percent} logic
(balancing liveness and safety risks),
as do most example policies we discuss in \cref{sec:qsc}.

\end{document}